\documentclass[suppldata]{interact}
\usepackage[english]{babel}
\usepackage[utf8]{inputenc}
\usepackage[sc]{mathpazo} 
\usepackage[T1]{fontenc}
\usepackage{epstopdf}

\usepackage{hyperref}
\usepackage{amsmath}
\usepackage{amsfonts}
\usepackage{amssymb}
\usepackage{subfig}
\usepackage{caption}
\usepackage{tabularx}
\usepackage[flushleft]{threeparttable}
\usepackage{multirow}
\usepackage{lscape}
\usepackage{enumitem} 
\setlist[itemize]{noitemsep} 

\usepackage[numbers,sort&compress]{natbib}
\bibpunct[, ]{[}{]}{,}{n}{,}{,}
\makeatletter
\def\NAT@def@citea{\def\@citea{\NAT@separator}}
\makeatother

\theoremstyle{plain}
\newtheorem{theorem}{Theorem}[section]
\newtheorem{lemma}[theorem]{Lemma}

\theoremstyle{definition}

\theoremstyle{remark}

\date{ }
\begin{document}

\articletype{}

\title{General framework for testing Poisson-Voronoi assumption for real microstructures}

\author{
  \name{Martina~Vittorietti\thanks{CONTACT M.~ Vittorietti. Email: m.vittorietti@tudelft.nl}\textsuperscript{a,b}, Piet J.J.~Kok\textsuperscript{c}, Jilt~Sietsma\textsuperscript{d}, Wei~Li\textsuperscript{d}, Geurt~Jongbloed\textsuperscript{a}}
  \affil{\textsuperscript{a}Department of Applied Mathematics, Delft University of Technology, Mourik Broekmanweg 6, 2628 XE Delft, The Netherlands;\\ \textsuperscript{b}Materials Innovation Institute (M2i), Mourik Broekmanweg 6, 2628 XE Delft, The Netherlands; \\ \textsuperscript{c}Tata Steel, IJmuiden Technology Centre, Postbus 10.000, 1970 CA, IJmuiden, The Netherlands; \\ \textsuperscript{d}Department of Materials Science and Engineering, Delft University of Technology, Mekelweg 2, 2628 CD Delft, The Netherlands.
  }}

\maketitle

\begin{abstract}
Modeling microstructures is an interesting problem not just in Materials Science but also in Mathematics and Statistics.
The most basic model for steel microstructure is the Poisson-Voronoi diagram.
It has mathematically attractive properties and it has been used in the approximation of single phase steel microstructures.
The aim of this paper is to develop methods that can be used to test whether a real steel microstructure can be approximated by such a model.
Therefore, a general framework for testing the Poisson-Voronoi assumption based on images of 2D sections of real metals is set out.
Following two different approaches, according to the use or not of periodic boundary conditions, three different model tests are proposed.
The first two are based on the coefficient of variation and the cumulative distribution function of the cells area. The third exploits tools from to Topological Data Analysis, such as persistence landscapes.

\end{abstract}


\section{Introduction}
The problem of quantifying complicated and fascinating microstructures of materials like metals has been around for many years.
It is an important issue in Materials Science because modeling 3D microstructures and relating these models to specific properties of the metals can give rise to new kinds of metals with desired performance. Indeed, having a good model for the microstructure, simulations can be performed to generate `digital versions' of the microstructure and testing its properties, for instance mechanical properties, using yet other models that establish the relation between microstructural and mechanical properties. These simulations, approximating reality, allow the researcher to test material at relatively low cost and relatively fast, compared to real physical experiments. It is clear that an important and challenging statistical question to be answered is whether a specific model for a microstructure is adequate, given measured data.

In the tentative of answering  this last question several points need to be touched upon.
The  first point concerns the choice of a model. There exists a vast choice of models and among them, Voronoi diagrams have been extensively studied and used \cite{okabe09}. In particular, Poisson-Voronoi diagrams, only involving one nonnegative intensity parameter $\lambda$, represent the most basic case for modeling microstrucures. In fact, they are often used in applications involving single-phase steel \cite{okabe09, kumar94, lorzkrawietz91}. More sophisticated models have been proposed, but in this paper we will concentrate on the Poisson-Voronoi model.

A second point concerns the available data. While the microstructure of a material is the arrangement of grains and
phases in a three dimensional (3D) space, the material is typically observed in two dimensions (2D). Usually, a small sample from inside the material is obtained and the exposed surface is examined in a microscope. Therefore, the work involves the study of 2D sections from which 3D microstructure information has to be extracted.

Under the 3D Poisson-Voronoi model, the observable 2D section is a realization of a so called \textit{2D sectional Poisson-Voronoi diagram}, often denoted by $\mathcal{V}_{\Phi}(2,3)$. It is the result of the intersection of a fixed plane and a 3D Poisson-Voronoi diagram. Only limited results about the geometrical characteristics of its grains have been obtained analytically but for most of them numerical results have been obtained through Monte Carlo simulations \cite{lorz90}. If a Poisson-Voronoi diagram is a good model, using 2D sections for the estimate of the intensity parameter $\lambda$, it is possible to infer distributions of almost all 3D microstructural properties, such as grain volume, grain surface area and grain number of faces \cite{vittorietti17}.

The last point is about the model validation. The question that this paper wants to answer is ``Given a real 2D materials section, could a Poisson-Voronoi diagram be a good model for approximating the 3D materials microstructure?'' We propose several tests for the Poisson-Voronoi hypothesis. These are all based on contrasts between features of the observed 2D picture and the features one would expect if the data were generated according to the Poisson-Voronoi model.




The paper structure is as follows.
After having reviewed the basic concepts of Voronoi diagrams (Section \textbf{\ref{sec:PVdiagrams}}), we recall the main stereological relations which can be used to estimate $\lambda$ based on a 2D sectional Poisson-Voronoi diagram and the most used intensity estimators introduced in \cite{lorzhahn94} (Section \textbf{\ref{sec:estimators}}).

Then, we move to the testing framework (Section \textbf{\ref{sec:test}}).
We distinguish periodic and non periodic boundary conditions. The former case is very popular in material science practice and it allows to approximate `infinite structures', giving nice scaling properties and  avoiding so called `edge effects'.
The latter more closely resembles real situations.
Assuming periodic boundary conditions, in Section \textbf{\ref{subsec:testper}} the distributions of the main geometrical characteristics of the 2D sectional cells are numerically obtained and two model tests are proposed.
The first, already introduced in \cite{lorzhahn93}, is based on the coefficient of variation of the cell (or grain) areas; the second is a Kolmogorov-Smirnov type test based on the cumulative distribution function of the cell areas.
In Section \textbf{\ref{subsec:testnonper}} the two tests previously mentioned are adapted to the non periodic boundaries setting. An additional test is defined, using tools from the emergent area of Topological Data Analysis (TDA), that combines the two disciplines of Statistics and Topology.
The focus is on \textit{persistent homology}, the branch of TDA that summarizes the 2D picture using various functions.
After having briefly and intuitively explained the basic concepts of persistent homology and the common ways of representing it (\textit{persistence diagram}), a  test based on the squared distances between \textit{persistence landscapes} is presented (Section \textbf{\ref{subsubsec:perapproach}}).

In Section \textbf{\ref{sec:quantile}}, we carry out a computer simulation for estimating the quantiles for the proposed model test statistics.
We consider null distributions for the test statistics conditional on the number of visible cells in 2D.
 For a general test statistic, the conditional distribution is expressed in terms of quantities that involve the (unknown) intensity parameter $\lambda$ of the 3D Poisson process and quantities independent of $\lambda$.

Finally, in Section \textbf{\ref{sec:application}}, we show an application of our work based on scanned images by \cite{lorzhahn93} of Alumina Ceramics. The different tests belonging to the different approaches are performed and the results are compared.
A brief discussion on future developments follows in Section \textbf{\ref{sec:discussion}}.

\section{Voronoi Diagrams}
\label{sec:PVdiagrams}
We begin reviewing the generic definition and the basic properties of the Voronoi Diagram.
Given a denumerable set of distinct points in $\mathbb{R}^d$, $\textbf{\textrm{X}}=\{x_i: i\ge1\}$, the Voronoi diagram of $\mathbb{R}^d$ with \textit{nuclei} $\{x_i\}$ (also called \textit{sites} or \textit{generator points}) is a partition of $\mathbb{R}^d$ consisting of cells
  \begin{equation}\notag
   C_i=\{ y\in\mathbb{R}^d\,:\, \| x_i-y \| \le \|x_j-y\|\,\, for\,\, j\ne i \},\,\,\, i=1,2,\dots
   \label{eq:vorcell}
   \end{equation}
   where $\|\cdot\|$ is the usual Euclidean norm.
   This means that given a set of two or more but finitely many distinct points, all locations in that space are associated with the closest member(s) of the point set with respect to the Euclidean distance.

If $\textbf{\textrm{X}}=\mathrm{\Phi}=\{x_i\}$ is the realization of a homogeneous Poisson point process, then we will refer to the resulting structure as the \textit{Poisson-Voronoi diagram} and denote it by $\mathcal{V}_\Phi$. This model is characterized by one single intensity parameter $\lambda$, the mean number of points generated according to the Poisson point process per unit volume.

Okabe et al. \cite{okabe09} synthesize previous research activity on the properties of Poisson-Voronoi diagrams.
Despite the fact that the moments of several geometrical characteristics are known, the distributions of the main features, especially in 3D, are not. In \cite{vittorietti17} a simulation study is conducted for finding accurate approximations for these distributions. A Generalized Gamma distribution is found to be the best approximating distribution among the well-known parametric densities frequently used in this framework.
Exploiting the scaling property of the Poisson process, one obtains the distribution of the main geometrical characteristics for $\lambda$.
In real experiments, it is often not possible to deal directly with 3D structures. Instead, one has to base inference on pictures of 2D sections of the 3D structure.
In \cite{chiu96}, Chiu et al.\ answer a fundamental question: ``For integers $2\le t\le d-1$, is the intersection between an arbitrary but fixed $t$-dimensional linear affine subspace of $\mathbb{R}^d$ and the $d$-dimensional Voronoi tessellation generated by a point process $\Phi$ a $t-$dimensional Voronoi tessellation?"
The answer to this question is negative when $\Phi$ is a Poisson point process \cite{mecke84, chiu96}. Moreover, each cell in a sectional Poisson-Voronoi tessellation is almost surely a \textit{non-Voronoi} cell \cite{chiu96}.
For 2D and 3D Poisson-Voronoi diagrams, also for 2D sectional Poisson-Voronoi diagrams, much information about moments and scaling for the main geometrical characteristics is known, but little information and no analytic expressions for their distributions of them are available so far.
In this paper, we focus on the distribution of the area, the perimeter and the number of edges of cells in 2D sectional Voronoi diagrams. The major results are summarized in Table \textbf{\ref{tab:moments}}.

\begin{table}[!h]
\centering
\caption{The first and second order moments of the main geometrical characteristics of a 2-dimensional sectional Poisson-Voronoi diagram}
\begin{threeparttable}
\begin{tabular}{lcc}
  \hline
   & Expected value &  Second moment \\
  \hline
 Number of vertices/edges & $6$&$38.827$\tnote{*} \\
 Area & $0.686\lambda^{-2/3}$&$0.699\lambda^{-4/3}$  \\
Perimeter & $3.136\lambda^{-1/3}$&$11.308\lambda^{-2/3}$\tnote{*}  \\
   \hline
\end{tabular}
\begin{tablenotes}
  \item[*] \footnotesize{The constants are estimated values; \cite{okabe09}.}
  \end{tablenotes}
  \end{threeparttable}
\label{tab:moments}
\end{table}

In the next section, we will see how stereological relations can be used to obtain estimates of the intensity parameter $\lambda$  of the 3D generating Poisson process based on the 2D sections.

\section{Stereological estimators for the intensity parameter $\lambda$}
\label{sec:estimators}
Basic stereological relationships exist which are independent of any underlying tessellation model. 
Moreover, in the literature explicit (scaling) relations are known expressing the expected number of vertices per unit area, $P_A$, the expected number of cells per unit area, $N_A$,  and the mean total edge length per unit area, $L_A$, in terms of the intensity parameter $\lambda$ for a generating 3D Poisson process. Combining stereological and scaling relationships, the following expressions hold, see e.g.\ \cite{lorzhahn94}. 
\begin{align*}\notag
P_A&=\frac{8}{15}\cdot\left(\frac{3}{4}\right)^{1/3}\cdot\pi^{5/3}\Gamma\left(\frac{4}{3}\right)\cdot\lambda^{2/3}=c_1\cdot\lambda^{2/3}\\
N_A&=\frac{4}{15}\cdot\left(\frac{3}{4}\right)^{1/3}\cdot\pi^{5/3}\Gamma\left(\frac{4}{3}\right)\cdot\lambda^{2/3}=\frac{c_1}{2}\cdot\lambda^{2/3} \mbox{ and }\\
L_A&=\pi\cdot\left(\frac{\pi}{6}\right)^{1/3}\cdot\Gamma\left(\frac{5}{3}\right)\cdot\lambda^{1/3}=c_2\cdot\lambda^{1/3}.
\end{align*}
Furthermore, exploiting the simple relation between $N_A$ and the expected area of the cell profiles, $\mathbb{E}(a)$,
$N_A=\frac{1}{\mathbb{E}(a)}$,
four estimators for $\lambda$ can be obtained:
\begin{align}
\notag
\hat{\lambda}_P&=\biggl(\frac{\hat{P}_A}{c_1}\biggr)^{3/2}\approx 0.2008\cdot\hat{P}_A^{3/2},\,\,
\hat{\lambda}_N=\biggl(\frac{2\hat{N}_A}{c_1}\biggr)^{3/2}\approx 0.5680\cdot\hat{N}_A^{3/2}\\\notag
\hat{\lambda}_L&=\biggl(\frac{\hat{L}_A}{c_2}\biggr)^{3}\approx 0.0837\cdot\hat{L}_A^{3},\,\, \label{eq:estimators}
\hat{\lambda}_a=\biggl(\frac{2}{c_1\bar{a}}\biggr)^{3/2}\approx 0.5680\cdot\bar{a}^{-3/2}.\\
\end{align}
Here the hats indicate natural estimates for the mean quantities based on the data (like `number of cells divided by observed area' for $\hat{\lambda}_n$).
In \cite{lorzhahn94}, the behavior of the estimators is investigated by means of a computer simulation. The authors state that the estimators show hardly any difference concerning bias and variance and that the biases are less than 1\% for sample size $n=50$ and that they decrease rapidly with increasing sample size.

Once we have an estimate of the intensity parameter $\hat{\lambda}$, it can be used for estimating the distribution of the main geometrical 3D features of the grains \cite{vittorietti17}.
An additional important issue of interest is whether the Poisson-Voronoi diagrams assumption is suitable in view of the observed 2D picture. We will consider this problem in the next section.

\section{Model tests for validity of the Poisson-Voronoi assumption}
\label{sec:test}
In \cite{lorzkrawietz91,lorzhahn93,stamm97} several model tests based on the distribution of geometrical features of the grains in random plane sections of a spatial tessellation are proposed.
More precisely in \cite{lorzkrawietz91}, the authors propose five stereological model tests based on the distribution of the number of cell vertices. The power of the model tests is investigated under some special parametric alternative hypotheses: a Mat\'ern cluster point process, a Mat\'ern hard-core point process and a simple sequential inhibition point process.
Secondly, in \cite{lorzhahn93} three different model tests are considered: the first two are based on the variability of the section cells area, the third is motivated by a well-known relationship between specific edge length $L_A$  and point process intensities $\lambda$ and $P_A$. In  line with their previous work, the authors propose one-sided and two-sided tests for distinguishing Poisson-Voronoi tessellation from more regular or irregular tessellations. The null distributions of the test statistics are approximated using simulation. Simulations also show that the model tests are quite powerful in discriminating the different kind of plane sections. It is interesting to note that all their tests are based on summarizing indices like the coefficient of variation, skewness index etc, and that the best behavior among them is reported to be the one based on the coefficient of variation of the cells area (eq. \textbf{\ref{eq:cv}}), also used by the authors in \cite{stamm97}.

In this paper we introduce test statistics that use more information contained in the data than only summarizing indices. To this end, we use tools belonging to different branches in statistics. Moreover, we describe a partly simulation-based framework to approximate null distributions of the test statistics considered.

Before going more deeply into the testing problem, it is necessary to make a distinction between periodic and non periodic boundary conditions.
On the one hand, periodic boundary conditions are mathematically convenient as these provide a natural way to deal with edge effects. Moreover, for large volumes and large values of $\lambda$, the construction really mimics the infinite volume situation where the convenient scaling results as mentioned in section \textbf{\ref{sec:estimators}}.  For real materials, the periodic boundary constraint is not realistic.
The approach without periodic boundary conditions is more realistic. It will be seen that determining null distributions of test statistics, the approach will be slightly more simulation based, but also more tailored to the data and 3D object at hand.

\subsection{Periodic Boundary Conditions}
\label{subsec:testper}
The first simulation study involves a Monte-Carlo procedure. The following results are obtained by randomly generating approximately $1\,000$ points in a box of dimension $10\times10\times10$ and using eq. \textbf{\ref{eq:vorcell}} for creating 3D Poisson-Voronoi cells. This is equivalent saying that the generator points of the Poisson-Voronoi diagram are generated according to a Poisson process with intensity parameter $\lambda=1$. Then, sections with dimensions $10\times10$ (parallel to the cube face for reducing boundary effect) are randomly taken. On average the number of 2D cells in a section turns out to be  approximately $146$.
The simulation is conducted using the software provided by TATA Steel.
The algorithm, that the software exploits is described in \cite{vittorietti17}.
The procedure consists of three main steps:\\



Repeat $1\,000\,000$ times:
 \begin{description}
\item[Step 1]: Generate a 3D Poisson-Voronoi diagram with intensity parameter $\lambda=1$ applying periodic boundary conditions;
\item[Step 2]: Take a random 2D section of the 3D structure;
\item[Step 3]: Determine the geometrical characteristics of all cells in the 2D section.
\end{description}

Graphical representations of the results are shown in Figures \textbf{\ref{fig:areaDensDistr}, \ref{fig:perDensDistr}, \ref{fig:edgesDensDistr}}.
For the grain area and the grain perimeter distributions estimation a simple boundary correction for kernel density estimation is adopted \cite{jones93}.
In fact, given the nonnegative support of the probability density functions of the two variables, the linear correction approach, as proposed in \cite{jones93}, prevents the estimate to assign mass outside$[0,\infty)$.

The values in Table \textbf{\ref{tab:moments}} are the estimated values of the main geometrical features for a 2D sectional Poisson-Voronoi diagram. They are in agreement with both theoretical and simulation results known in the literature (cf. \cite{okabe09}).

\begin{figure}[!h]
\centering
\subfloat[]{\includegraphics[width=7cm]{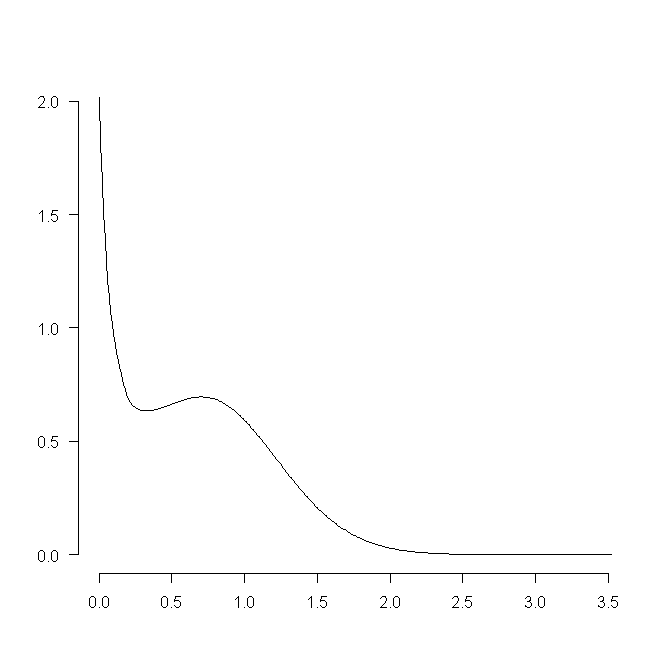}}
\subfloat[]{\includegraphics[width=7cm]{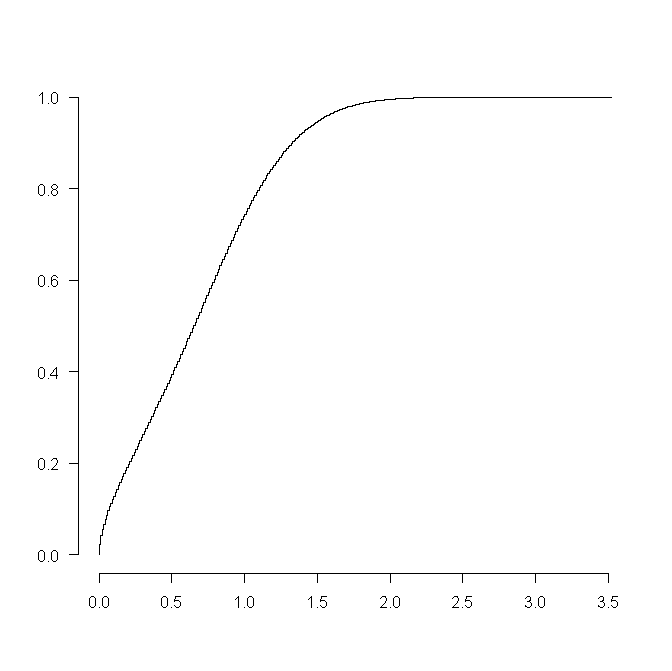}}
\caption{(a) Boundary corrected Kernel density estimate (Epanechnikov kernel, linear combination correction, $h=0.2$ \cite{jones93}) and (b) empirical cumulative distribution function of the area of 36\,480\,600 (originating from the 1\,000\,000 slices) 2D sectional cells, $\lambda=1$ }
\label{fig:areaDensDistr}
\end{figure}

\begin{figure}[!h]
\centering
\subfloat[]{\includegraphics[width=7cm]{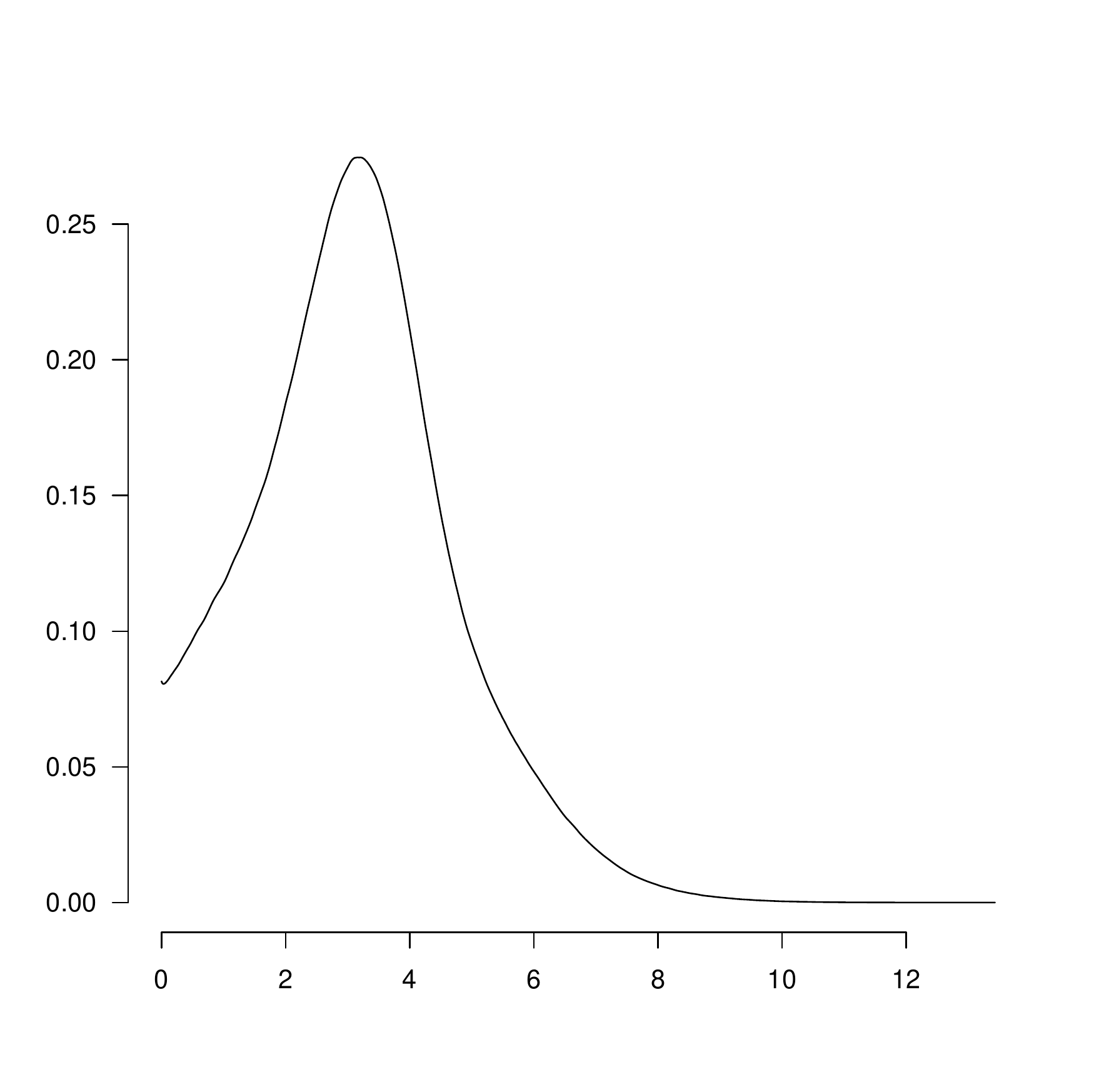}}
\subfloat[]{\includegraphics[width=7cm]{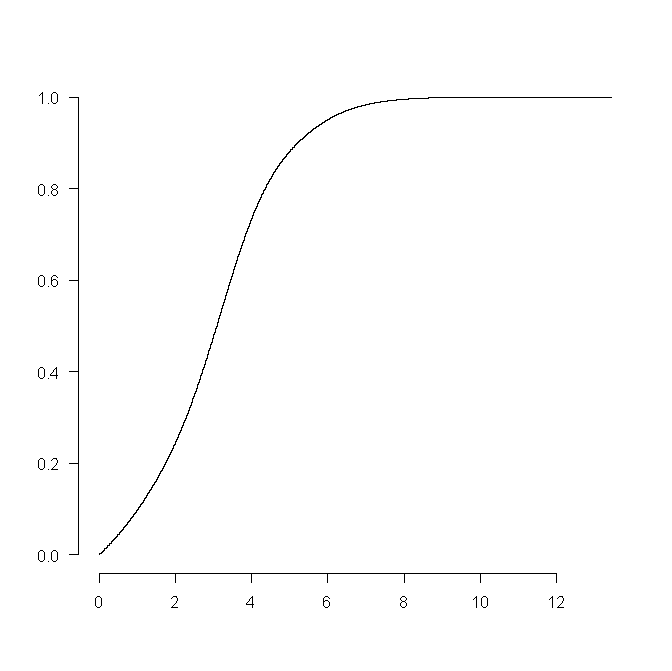}}
\caption{(a) Boundary corrected Kernel density estimate (Epanechnikov kernel, linear combination correction, $h=0.1$ \cite{jones93}) and (b) empirical cumulative distribution function of the perimeter of 36\,480\,600 2D sectional cells, $\lambda=1$ }
\label{fig:perDensDistr}
\end{figure}

\begin{figure}[!h]
\centering
\subfloat[]{\includegraphics[width=7cm]{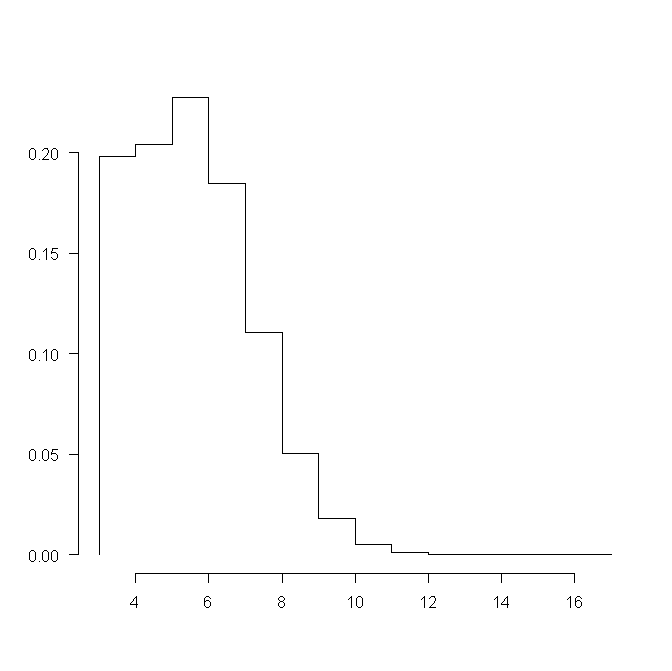}}
\subfloat[]{\includegraphics[width=7cm]{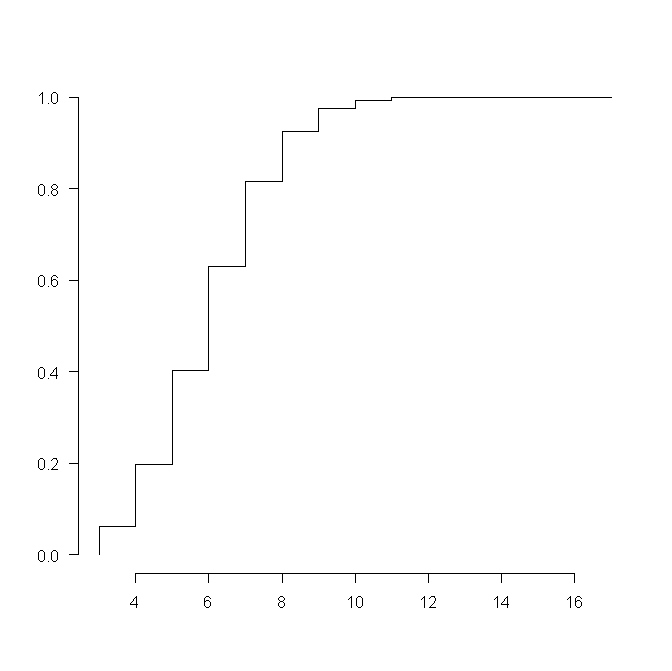}}
\caption{(a) Relative frequencies and (b) empirical cumulative distribution function of the number of edges of 36\,480\,600 2D sectional cells, $\lambda=1$}
\label{fig:edgesDensDistr}
\end{figure}

\begin{table}[!h]
\centering
\caption{Estimated moments of the geometrical features of 36\,480\,600 2D sectional cells, $\lambda$=1}
\subfloat[][\emph{Area}]
{\begin{tabular}{rr}
  \hline
 $\mu_1$ & 0.68524 \\
  $\sigma $& 0.47342 \\
 $ \mu_2 $& 0.69367 \\
  $\mu_3 $& 30.37169 \\
$  \mu_4$ & 40.94590 \\
   \hline
\end{tabular}}
\subfloat[][\emph{Perimeter}]
{\begin{tabular}{rr}
  \hline
$\mu_1$ & 3.13345 \\
$ \sigma$ & 1.60552 \\
$ \mu_2 $& 12.39622 \\
  $\mu_3 $& 2072.73503 \\
  $\mu_4$ & 10695.17596 \\
   \hline
\end{tabular}}
\subfloat[][\emph{Number of edges}]
{\begin{tabular}{rr}
  \hline
$\mu_1 $& 6.00000 \\
$  \sigma$ & 1.69195 \\
$  \mu_2 $& 38.86268 \\
  $\mu_3 $& 9818.30810 \\
 $ \mu_4 $& 72107.17324 \\
   \hline
\end{tabular}}
\label{tab:moments}
\end{table}

If it comes to the study of mechanical properties of metal, the grain size is known to be an important parameter. In 2D, grain area therefore represents  one of the most interesting features for real materials sections, especially for single-phase materials \cite{hermann1989}.
Therefore, in this paper we restrict ourselves to tests based on observed cell areas.
The first one, mentioned before and already used in \cite{lorzhahn93,stamm97}, is based on the coefficient of variation of the observed cell areas:

\begin{equation}
C=\frac{\sqrt{\frac{1}{n-1}\sum_{i=1}^n(a_i-\bar{a})^2}}{\bar{a}}.
\label{eq:cv}
\end{equation}

Here $a_i$ is the area of the $i$-th sectional cell and $\bar{a}$ is the mean cell area in the section.
As the coefficient of variation is scale invariant, one just needs to compute the coefficient of variation of the area of the cells of a real section applying periodic boundary conditions and compare it with the quantile of the distribution of this test statistic. In fact, the information contained in the 2D section is clearly related to the number of cells observed ($n$) and comparing the observed value of $C$ with a quantile of the conditional distribution of $C$ given $n$ will only depend on number of cells observed in the 2D section. 

The second test is based on the cumulative distribution function (CDF) of the area of the 2D sectional cells. More precisely it is a Kolmogorov-Smirnov type test given by the supremum distance between the CDF of the area of the cells of the section for which one wants to test the Poisson-Voronoi hypothesis and a function that reflects our expectation of the empirical distribution function under the Poisson-Voronoi assumption. For the latter we choose a very accurate simulation-based approximation of the CDF of the area of $36\,480\,600$ sectional Poisson-Voronoi cells.
Let $F_1$ be the cumulative distribution function of the areas of the 2D sectional cells with intensity parameter $\lambda=1$ approximated via simulation as described above and let $\hat G$ be the empirical distribution function of the area of $n$ cells of a 2D section from a 3D structure with intensity parameter $\lambda$. First, we use eq. \textbf{\ref{eq:estimators}} for estimating the intensity based on the considered section, $\hat{\lambda}_a$.
Furthermore, inspired by \textbf{Lemma 3} in \cite{vittorietti17}, we define the next test statistic as the supremum distance between the two functions:
\begin{equation}
D(F,\hat G)=\sup_{x\ge0}|F_1(x)-\hat G(\hat{\lambda}^{\frac{2}{3}}x)|.
\end{equation}
We will return to the issue of approximating the null distribution of this test statistic in Section \textbf{\ref{sec:quantile}}.

\subsection{Non Periodic Boundary Conditions}
\label{subsec:testnonper}
In most real situations, the data available are relative to a material section with completely visible as well as partially visible grains. In such situations it is not realistic to use periodic boundary conditions in the model. 
We fix the geometry of the 3D volume and 2D slice as in the periodic boundaries case (Section \textbf{\ref{subsec:testper}}). Then the procedure can be summarized in three main steps:\\
Repeat $1\,000\,000$ times:
 \begin{description}
\item[Step 1]: Generate a 3D Poisson-Voronoi diagram with intensity parameter $\lambda$ not applying periodic boundary conditions. In this paper, for reasons that will become clear later, $\lambda=0.2$ is chosen;
\item[Step 2]: Take a random 2D section of the 3D structure;
\item[Step 3]: Determine the geometrical characteristics of the completely visible and the partially visible cells in the 2D section.
\end{description}
In this setting we consider three different tests.

The first revokes exactly the one shown in the setting of Section \textbf{\ref{subsec:testper}} (eq.\textbf{ \ref{eq:cv}}) and it is based on the coefficient of variation of the area of the totally and partially visible cells. Obviously, the referring quantile of the distribution of the statistical test are different with respect to the previous case.
The second statistic is in line with the test based on the cumulative distribution function of the cell areas seen in Section \textbf{\ref{subsec:testper}}, but the formulation is slightly different. It is expressed by

\begin{equation}
D(\bar F_{\lambda \,n_{2D}},\hat G_{n_{2D}})=\sup_{x\ge0}|\bar F_{\lambda \,n_{2D}}(x)- \hat G_{n_{2D}}(x)|
\label{eq:testcdf}
\end{equation}
where $\bar F_{\lambda \,n_{2D}}$
is the expected CDF conditioned to the event of observing exactly $n_{2D}$ sectional cells with estimated parameter $\lambda$. In Section \textbf{\ref{sec:quantile}}, it will be explained in more detail how this can be computed.
$\hat G_{n_{2D}}$ is the CDF of the areas of the totally and partially visible cells of the section under study. 

The last test exploits tools coming from the emergent field of Topological Data Analysis. We will now explain the main concepts of persistence homology necessary for using our model test.

\subsubsection*{Test based on Persistence plots}
\label{subsubsec:perapproach}
Instead of giving rigid mathematical and topological definitions, the aim of this section is to guide the reader via intuitive concepts in the construction of \textit{persistence diagrams} and \textit{persistence landscapes} used for the last model test.
Looking at one 2D image, it is hard to identify the really `important' features that univocally characterize it. 
 Topological data analysis (TDA) is a relatively new discipline that has provided  new insight into the study of qualitative features of data. In particular, persistent homology is the branch of TDA that provides tools both for identifying qualitative features of data and to give a measure of the importance of those features.  Key topological features of a set include connected components, holes, voids \dots. The main aim of persistent homology is to record the evolution of those characteristics with respect to a scale parameter $r$ that usually can be interpreted as time.
 
For avoiding too long digressions that can drift away from the real scope of the paper, most of the main concepts belonging to homology and persistent homology field are just roughly mentioned. For readers that aim to come to a formal definition of the following procedure, in \cite{hatcher2002} more details are provided.



For illustrative reasons and because in this study 2D images are used, the 2D case is considered but generalization to higher dimensions is not complicated.
The input of the analysis typically takes the form of a point cloud $\textbf{\textrm{X}}$ (Fig. \textbf{\ref{fig:perexpl}(a)}).
Based on that, a special structure is built.
It provides information about the qualitative features above discussed.
This structure is based on  so called \textit{simplices}.
A \textit{geometric} $k$\textit{-simplex} is the convex hull of $k+1$ affinely independent points $v_0,v_1,\dots,v_k$.
More precisely the $0$-simplex identifies vertices, the $1$-simplex line segments and the $2$-simplex triangles.

\begin{figure}[!h]
\centering
\subfloat[]{\includegraphics[width=5cm]{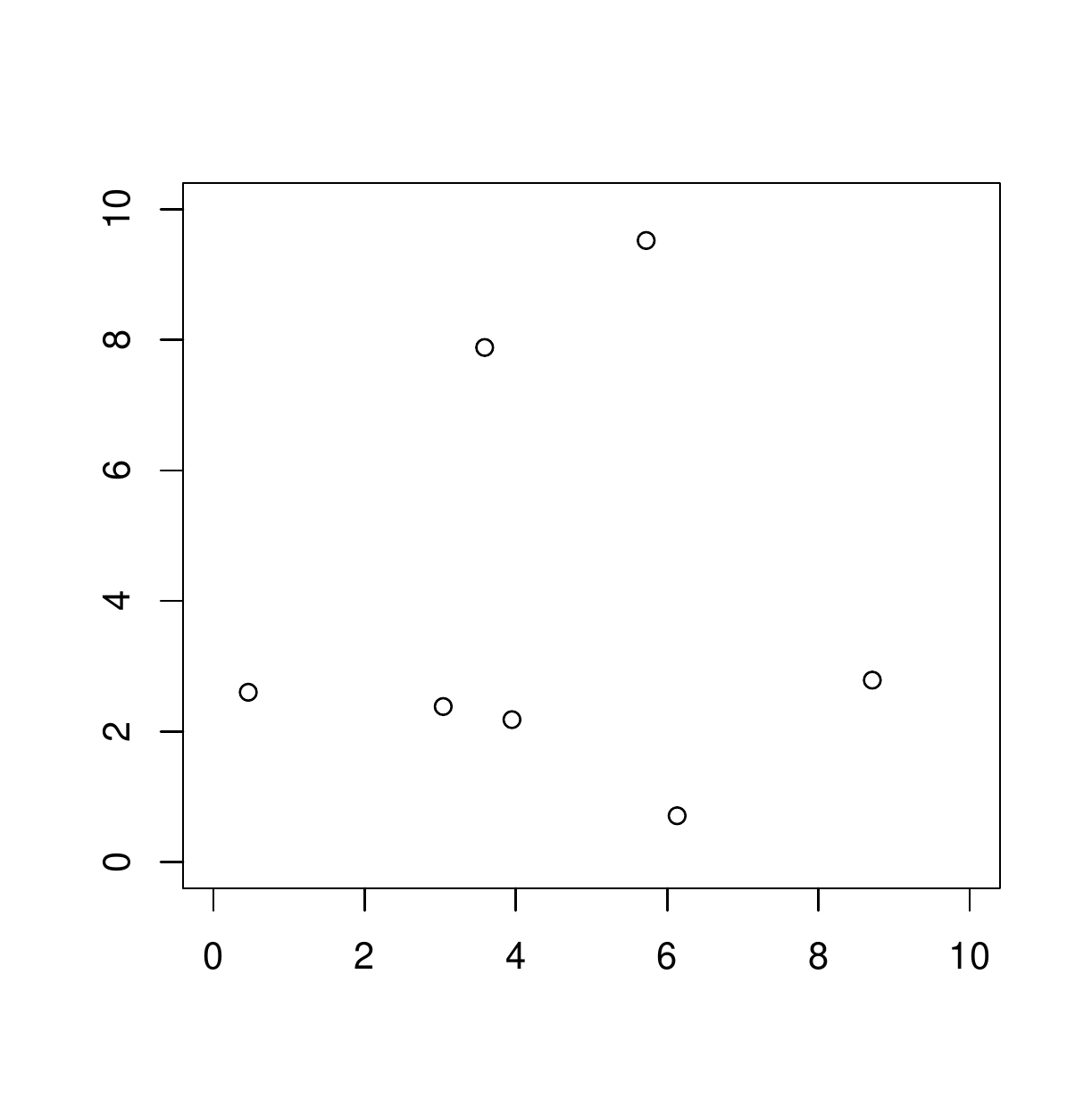}}
\subfloat[]{\includegraphics[width=5.5cm]{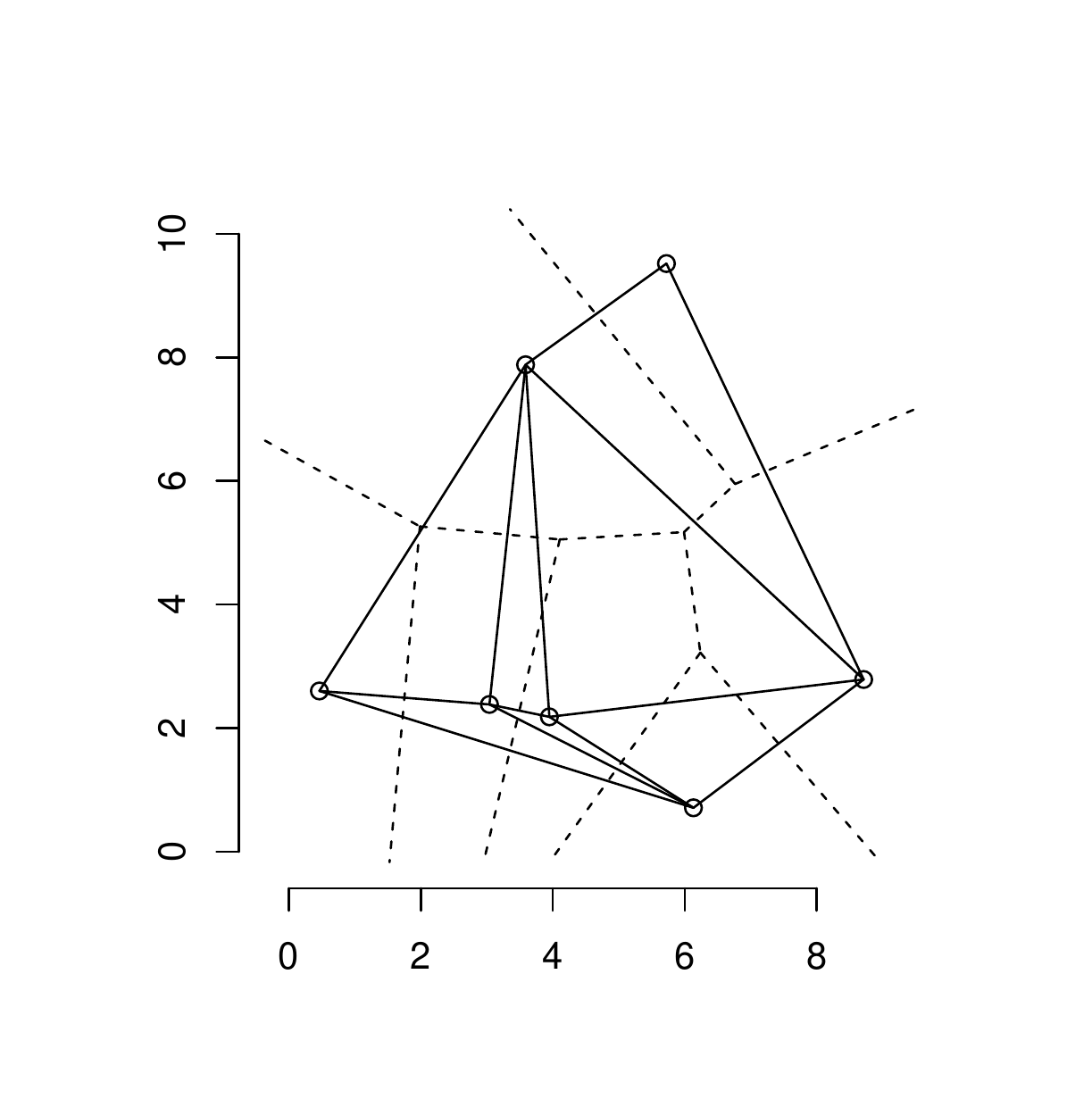}}
\subfloat[]{\includegraphics[width=5.5cm]{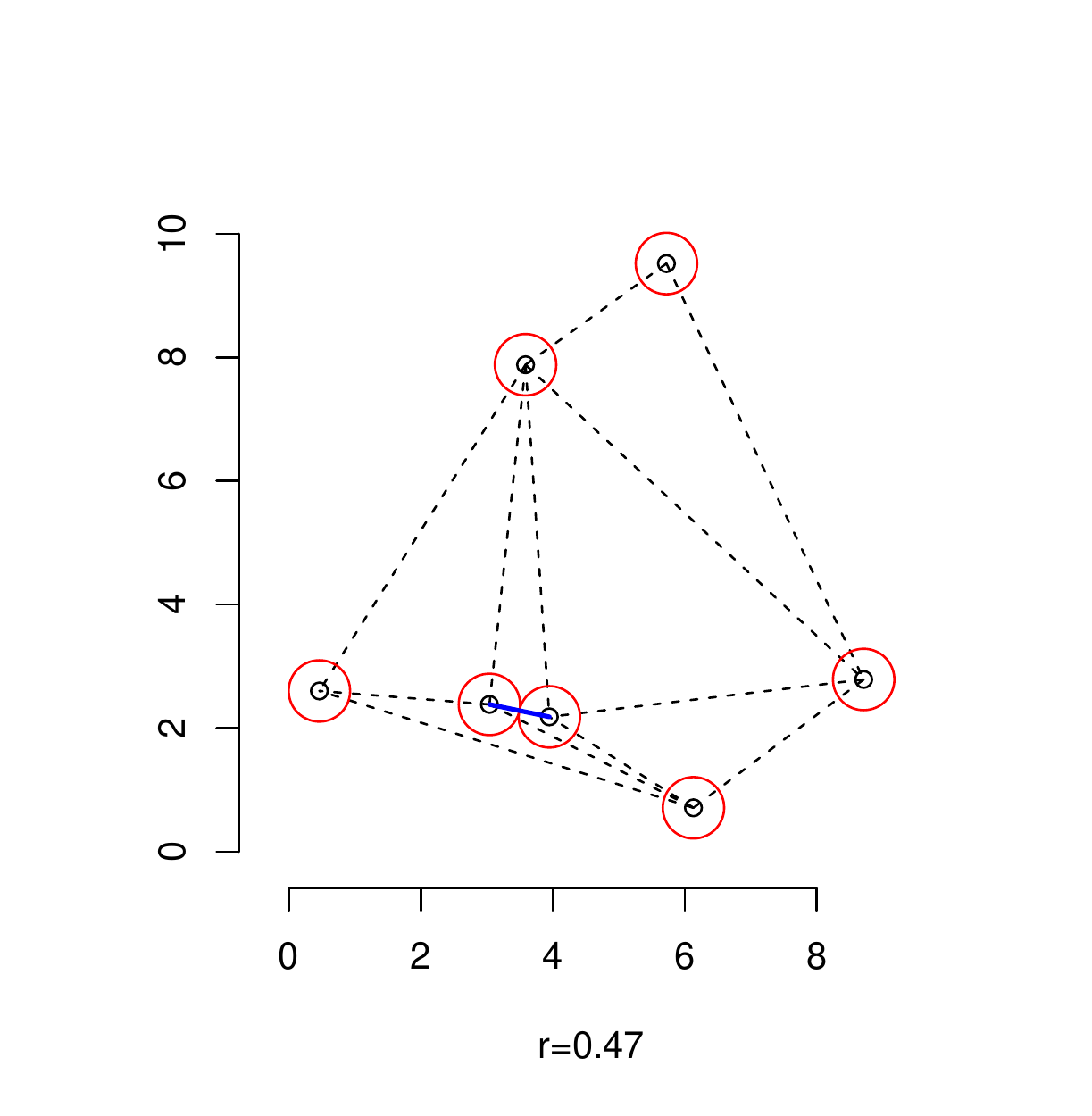}}

\subfloat[]{\includegraphics[width=5cm]{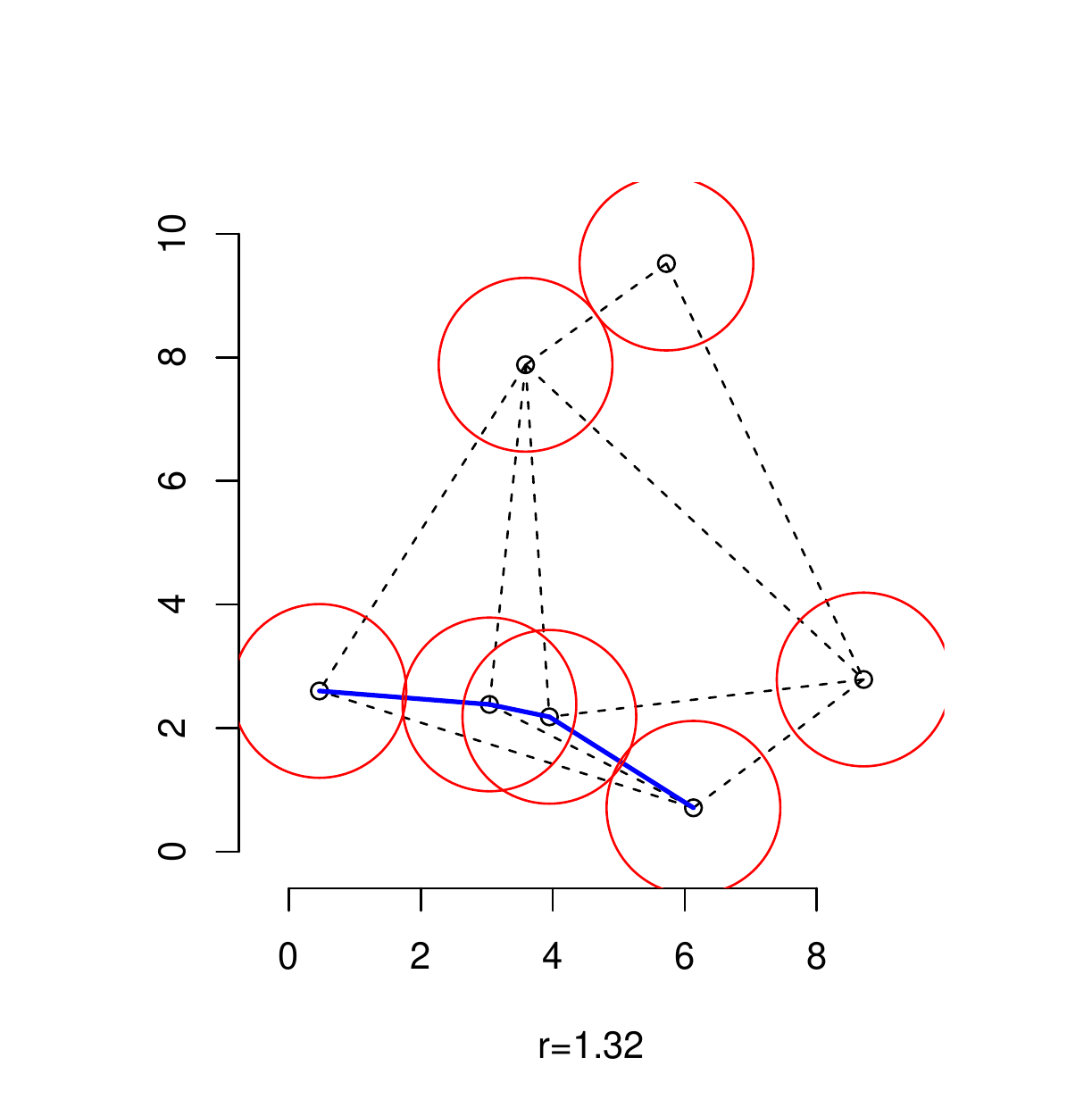}}
\subfloat[]{\includegraphics[width=5cm]{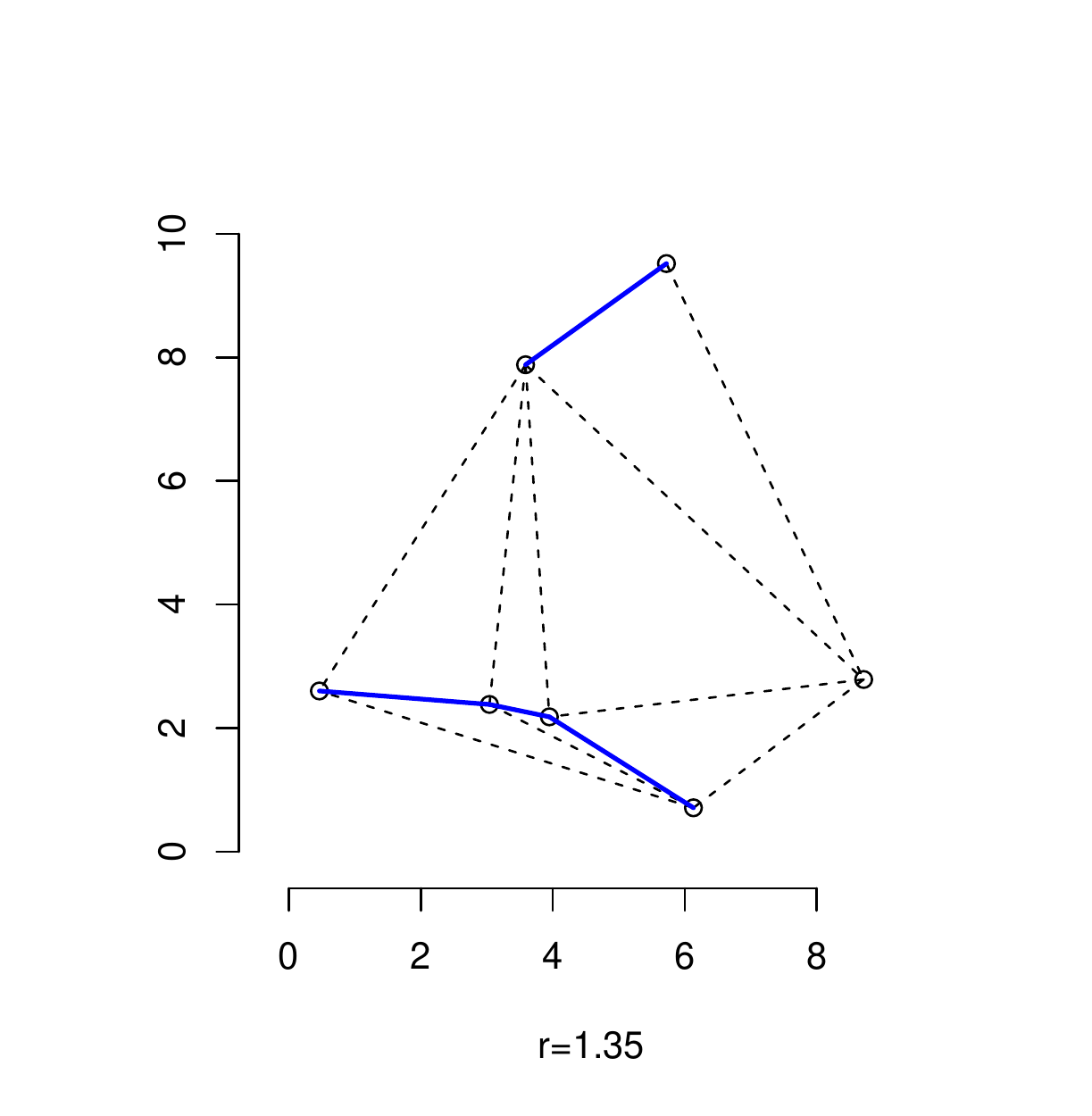}}
\subfloat[]{\includegraphics[width=5cm]{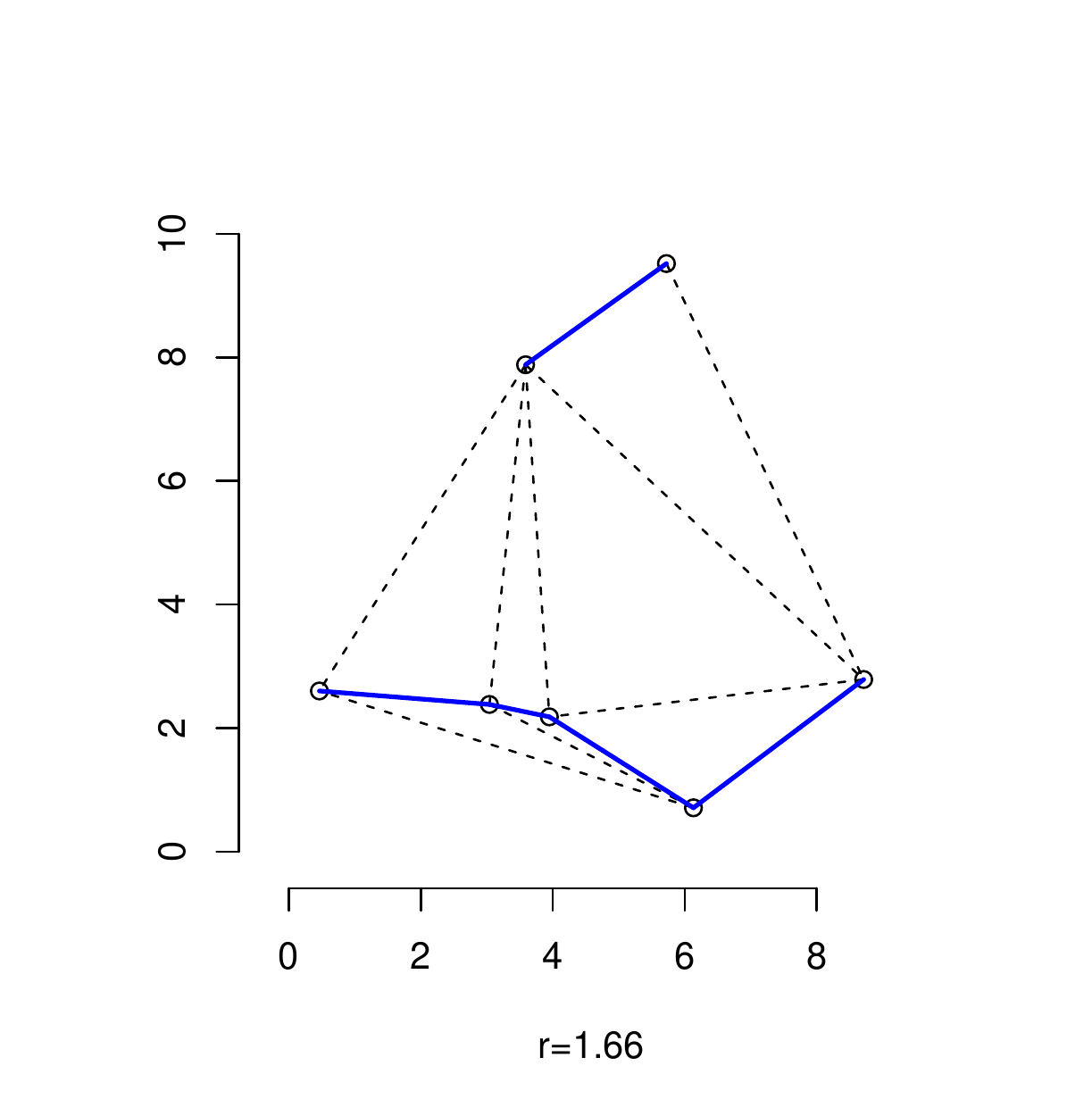}}

\subfloat[]{\includegraphics[width=5cm]{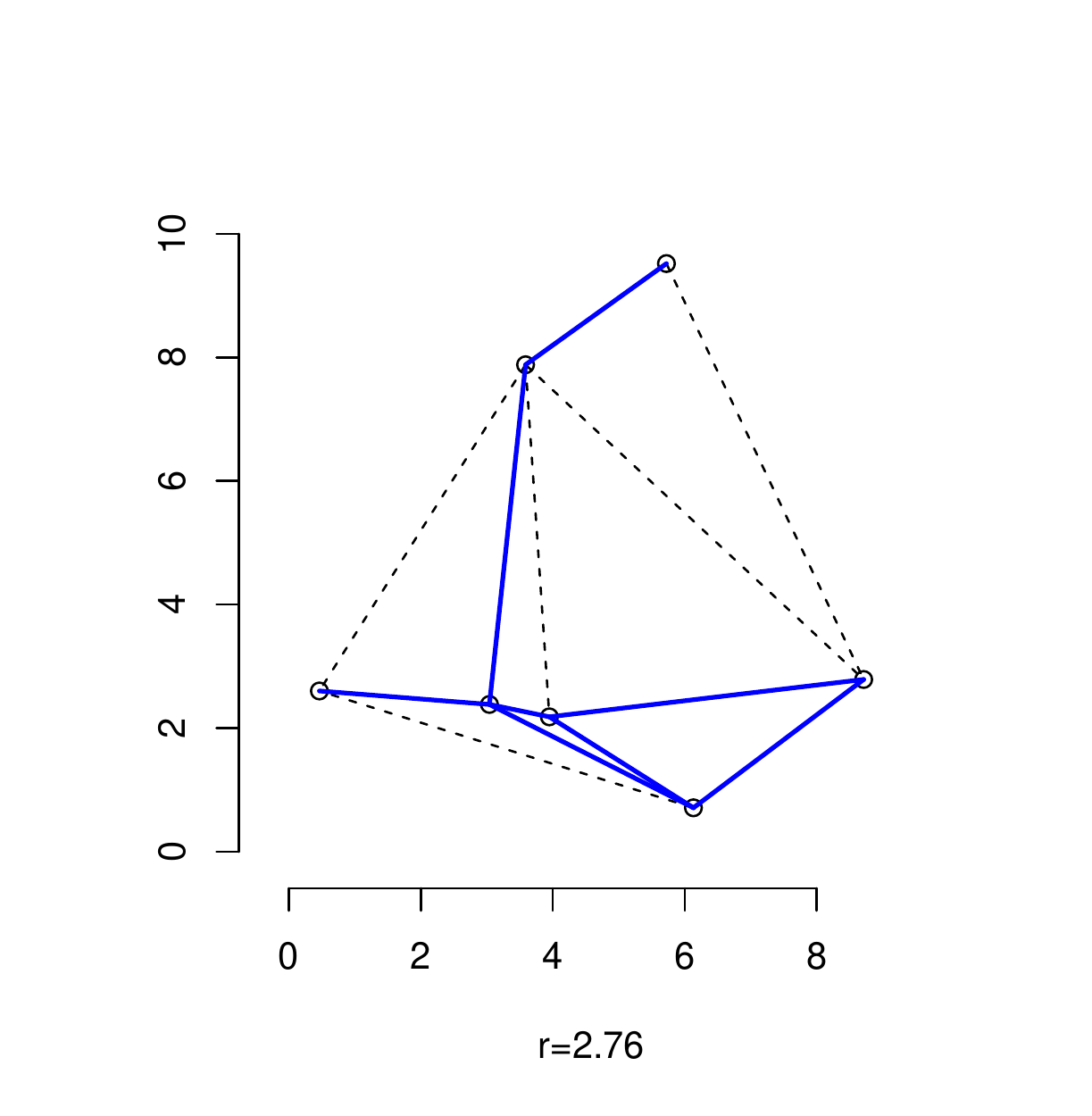}}
\subfloat[]{\includegraphics[width=5cm]{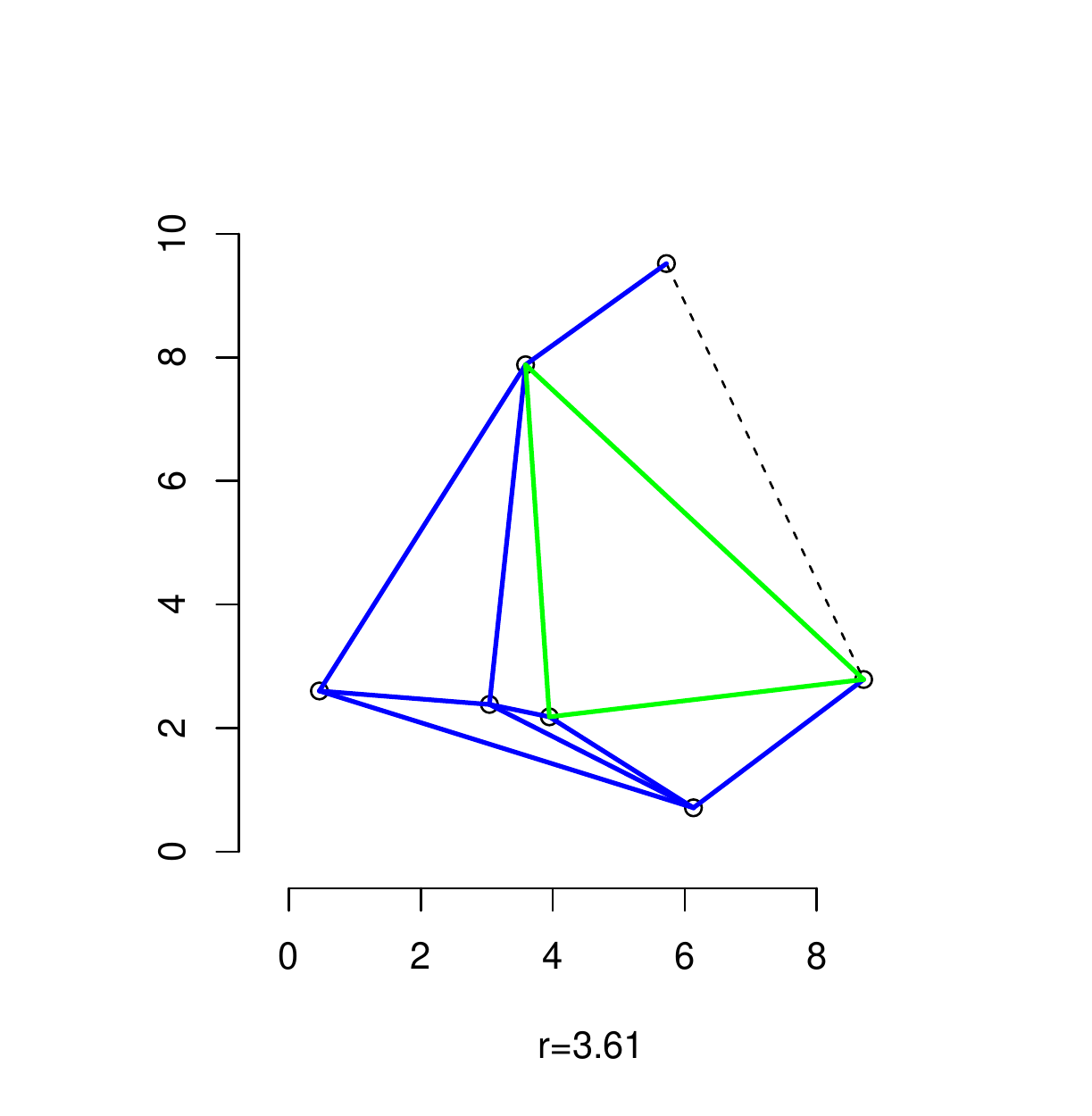}}
\subfloat[]{\includegraphics[width=5cm]{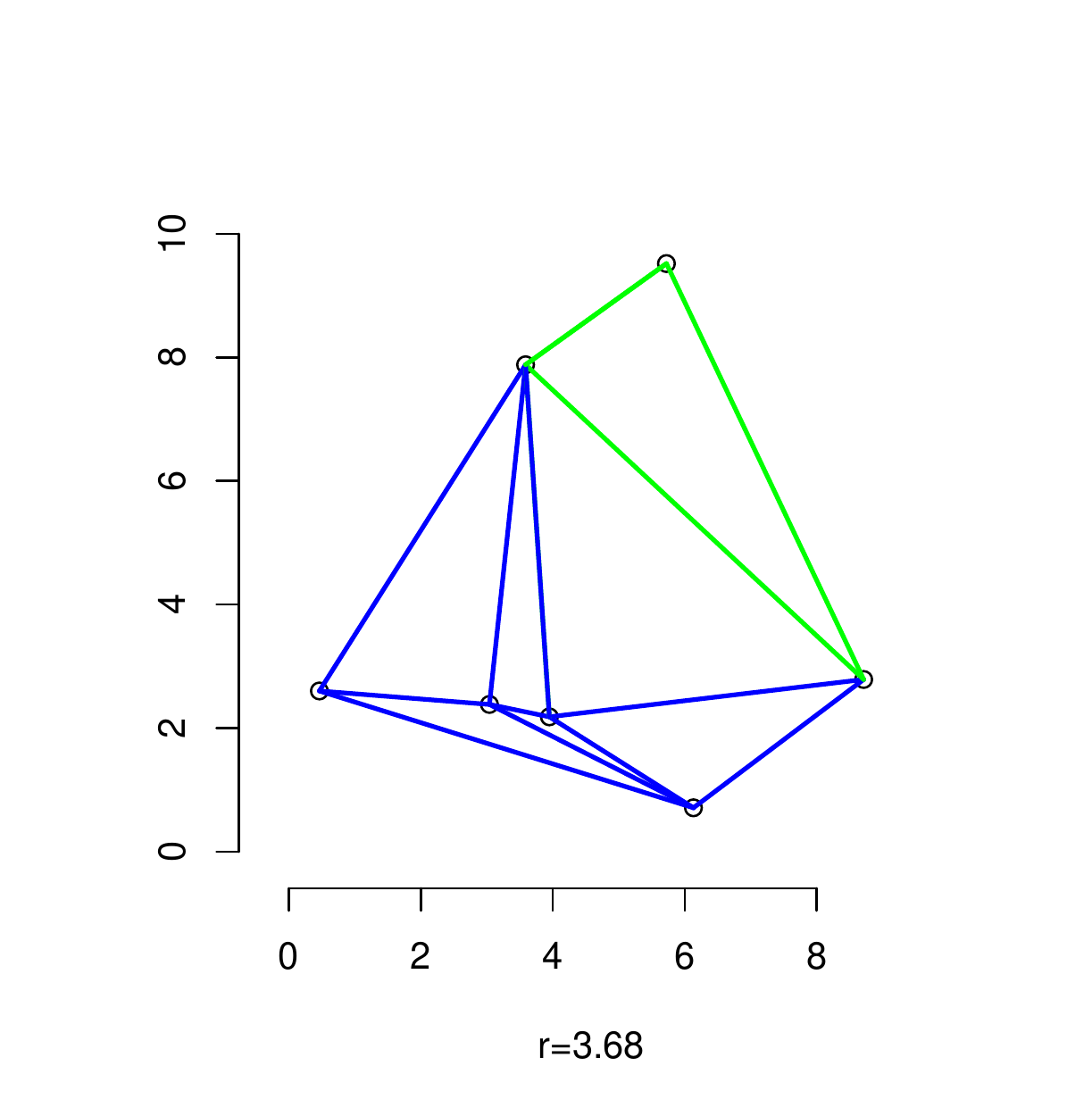}}

\caption{(a) Set of points $\textbf{\textrm{X}}$ (b) Voronoi Diagram (dashed) and Delaunay Triangulation (solid) (c) Circles with radius $0.47$ around the points of $\textbf{\textrm{X}}$; the Alpha complex $\alpha_r(\textbf{\textrm{X}})$ consists of the individual points of $\textbf{\textrm{X}}$ and the one edge corresponding to the two touching circles  (d)  Alpha complex for $r=1.32$ (e) Alpha complex, $r=1.35$ (f) Alpha complex, $r=1.66$ (g) Alpha complex, $r=2.76$ (h) Alpha complex, $r=3.61$ (i) Alpha complex, $r=3.68$; the entire Delaunay Triangulation.}
\label{fig:perexpl}
\end{figure}

One way for building this structure starts off with the 
so called \textit{Delaunay Triangulation}, $DT(\textbf{\textrm{X}})$ of $\textbf{\textrm{X}}$. Basically, this is a graph consisting of vertices in $\textbf{\textrm{X}}$ and edges between two points if and only if they share a Voronoi edge (Fig. \textbf{\ref{fig:perexpl}(b)}).  Then, circles are grown with increasing radius $r$, centered at the points in $\textbf{\textrm{X}}$.
The \textit{Alpha complex}\footnote{Other common choices of simplicial complexes are \v{C}ech  and Rips 
complexes; see \cite{hatcher2002,edelsbrunner2010}.} at radius $r$, $\alpha_r(\textbf{\textrm{X}})$, is a subcomplex of $DT(\textbf{\textrm{X}})$. In fact, for $r$ very small, the Alpha complex is nothing but the set $\textbf{\textrm{X}}$ of the generator points. Then $r$ grows and once two circles intersect, the edge of the underlying Delaunay triangulation between the two circle centers is added to $\alpha_r(\textbf{\textrm{X}})$. Eventually, for $r$ very big, the Alpha complex is the Delaunay Triangulation itself (Fig. \textbf{\ref{fig:perexpl}(c-i)}).

Now, rather than considering this structure for some fixed value of $r$, its evolution for growing $r>0$ is registered. In particular, we keep track of the birth time $b$ and a death time $d$ of connected components and holes\footnote{In three dimensions one could also consider other topological features, like loops or voids}, where the `time' is given by the radius of the circles corresponding to those events. One can think of the circles radii growing at constant rate. At time zero, the Alpha complex equals $\textbf{\textrm{X}}$. All individuals are connected components in themselves. These are born at time zero. After some time, when the first two points get connected because their circles touch, one can say two connected components merge or one connected component `dies'. In Figure \textbf{\ref{fig:perexpl}}, this happens for $r=0.47$; see subplot (c). For one connected component, we therefore have $(b,d)=(0,0.47)$. Increasing $r$ further, more connected components will `die' until only one remains for all $r$ large enough because all points are covered by the union of all large circles. During the same process, it is also possible that holes appear. This happens when a triangle appears in the picture, such that the $r$-circles around the three corner points of this triangle do not cover the whole triangle. At this time a hole is `born', yielding a birth time $b$ for this feature. It will also `die' again, when $r$ is further increased and the circles centered at the corners do cover the whole triangle. Note that not all triangles that appear correspond to the birth of a hole. For instance, in Figure \textbf{\ref{fig:perexpl}(g)} a triangle appears but the circles centered at the three corners immediately cover the whole triangle.

The points $(b,d)$ thus obtained can be used as coordinates and plotted on a plane, resulting in the so called \textit{persistence diagram}. Since the topological features (connected components, holes) can only die \textit{after} they are born ($d\ge b$), necessarily each point appears on or above the diagonal line $y = x$. The persistence diagram corresponding to the data in Figure \textbf{\ref{fig:perexpl}} is shown in Figure \textbf{\ref{fig:perexpl2}}. The black dots, $D_{0i}$, on the vertical axis represent the `deaths' of connected components; the lowest being the aforementioned $(b,d)=(0,0.47)$, the highest, $(b,d)=(0,2.67)$, corresponding to Figure \textbf{\ref{fig:perexpl} (g)}. The red triangles $D_{1i}$, represent the birth- and death times of the holes.

Based on persistence diagrams, several descriptive summarizing functions have been proposed in the literature. For example
rank functions \cite{robins2016}, landscapes and silhouettes \cite{bubenik2015, chazal2014} and accumulated persistence functions \cite{biscio2016}.
In this paper we follow the persistence landscapes approach, but any other summary statistic could also be used for testing.

We first describe in words how to construct a landscape from a persistence diagram. Then, the formal definition follows.
For each point $(b,d)$ in the persistence diagram, count the number of points to its left top (north-west). This is the \textit{rank} of the point $(b,d)$ and it can be interpreted as the number of features that are alive at time $b$ and that are still alive at time $d$. Then, draw horizontal and vertical lines from each point $(b,d)$ in the persistence diagram to the diagonal and `tip the diagram on its side'. Then take the contour of the projection of the points with the same rank. This results in the so-called landscape. This is done for connected components and holes separately, see Figure \textbf{\ref{fig:perrank}}.

\begin{figure}[!h]
\centering
\includegraphics[width=6cm]{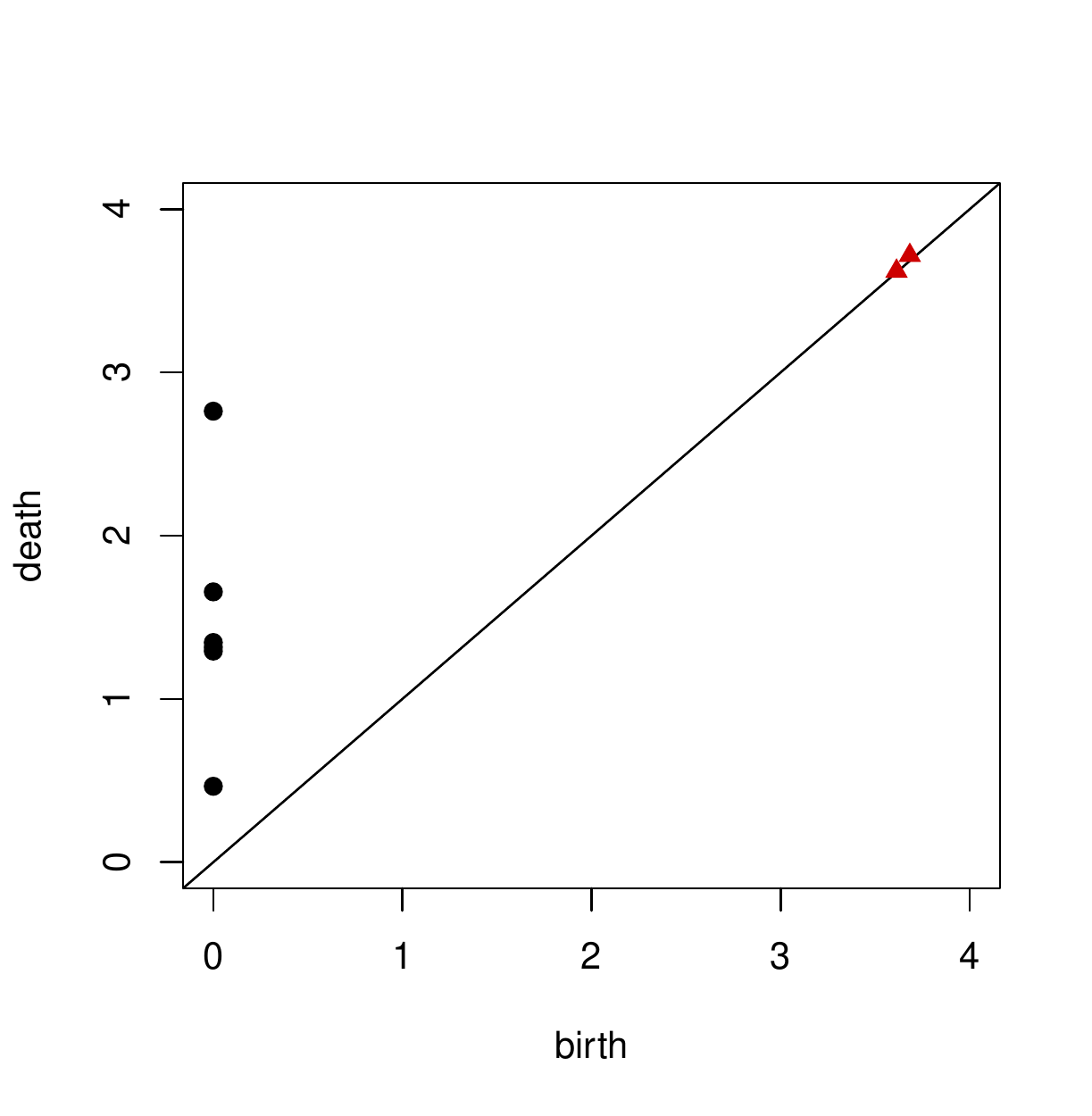}
\caption{Persistence Diagram. The black dots indicate the birth- and death time of connected components and the red triangles the birth- and death times of the holes. The data are the same as those used for Figure \textbf{\ref{fig:perexpl}}.}
\label{fig:perexpl2}
\end{figure}

\begin{figure}[!h]
\centering
\subfloat[]{\includegraphics[width=6cm,angle=-90]{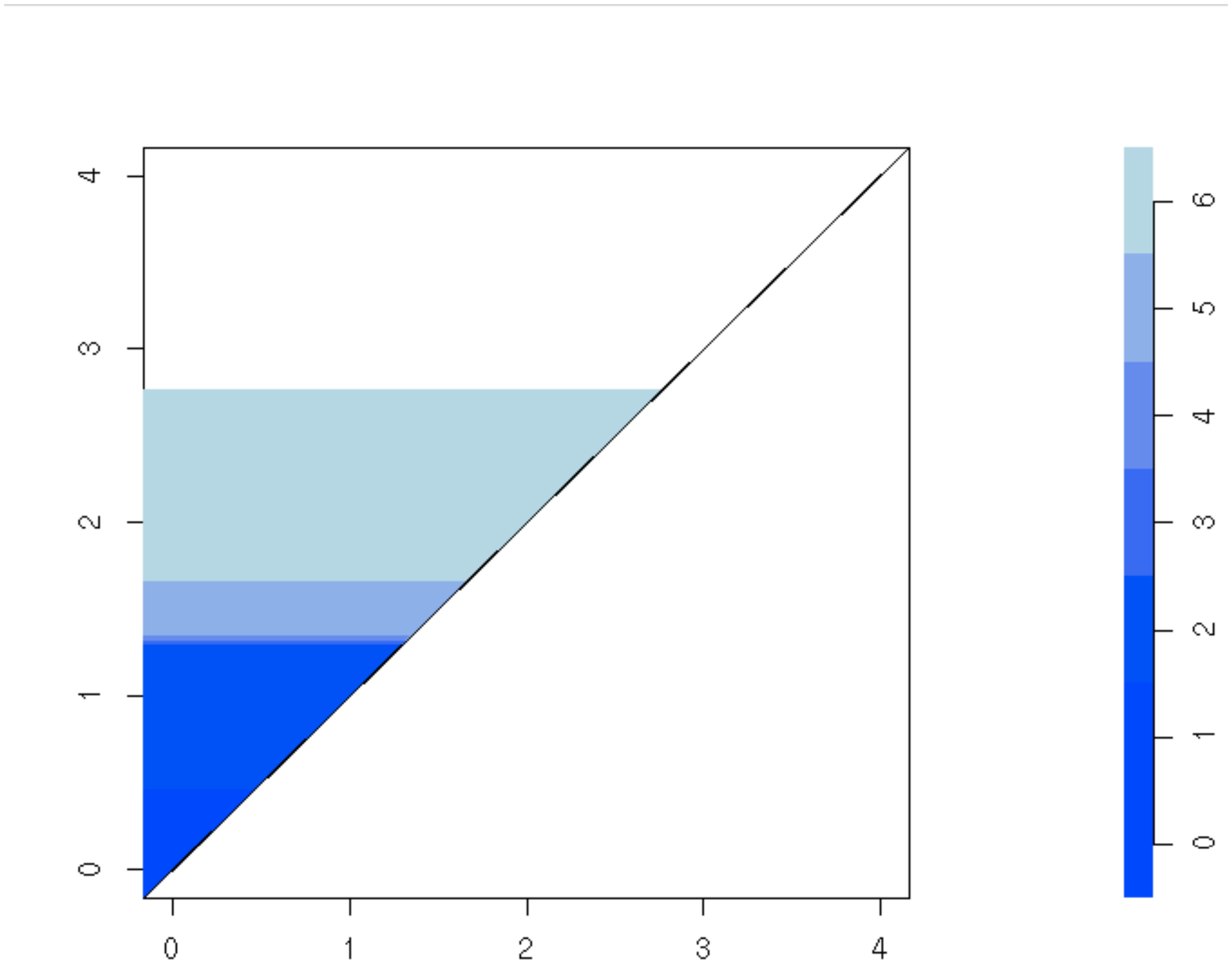}}
\subfloat[]{\includegraphics[width=6cm,angle=-90]{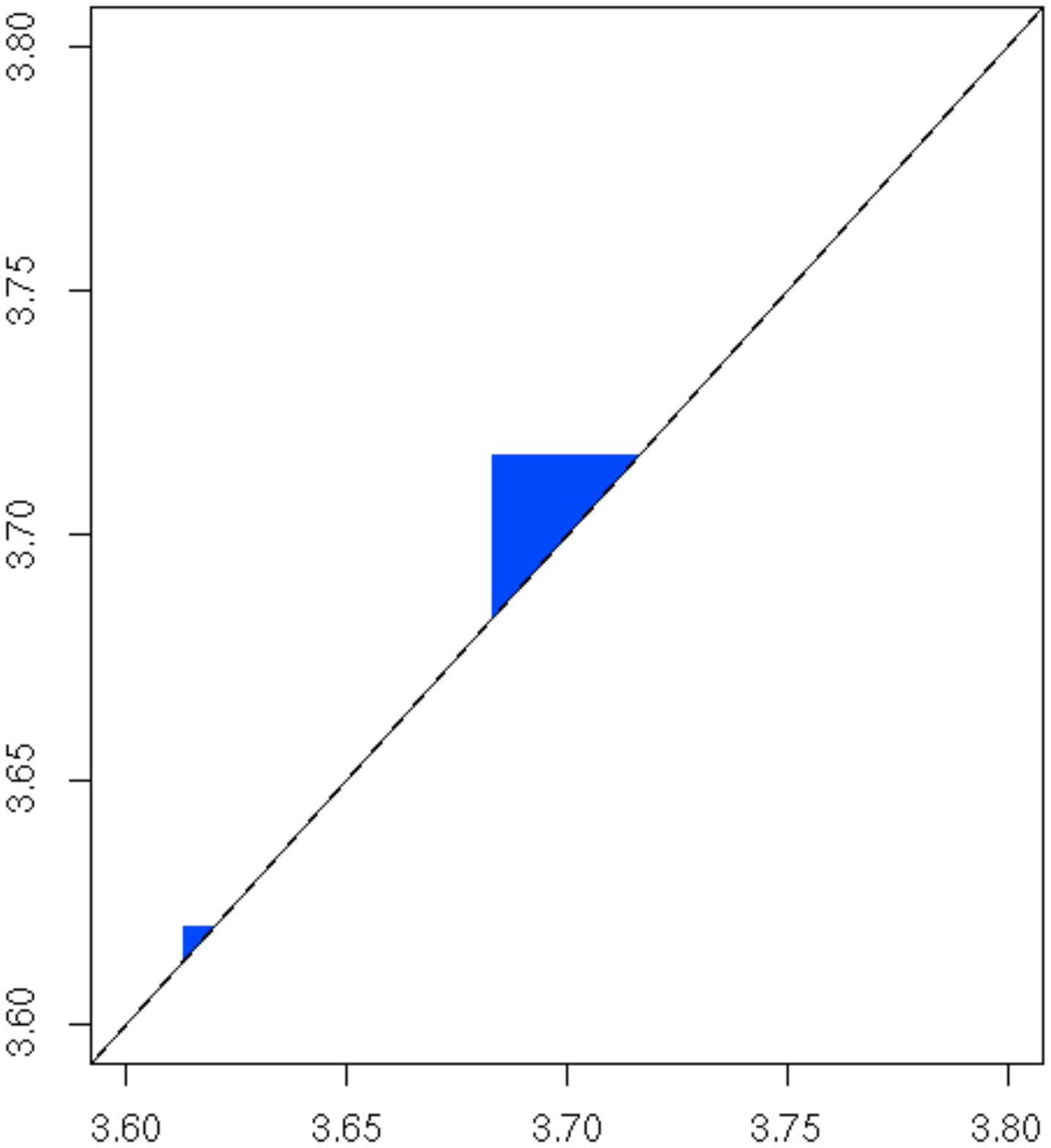}}

\subfloat[]{\includegraphics[width=7cm]{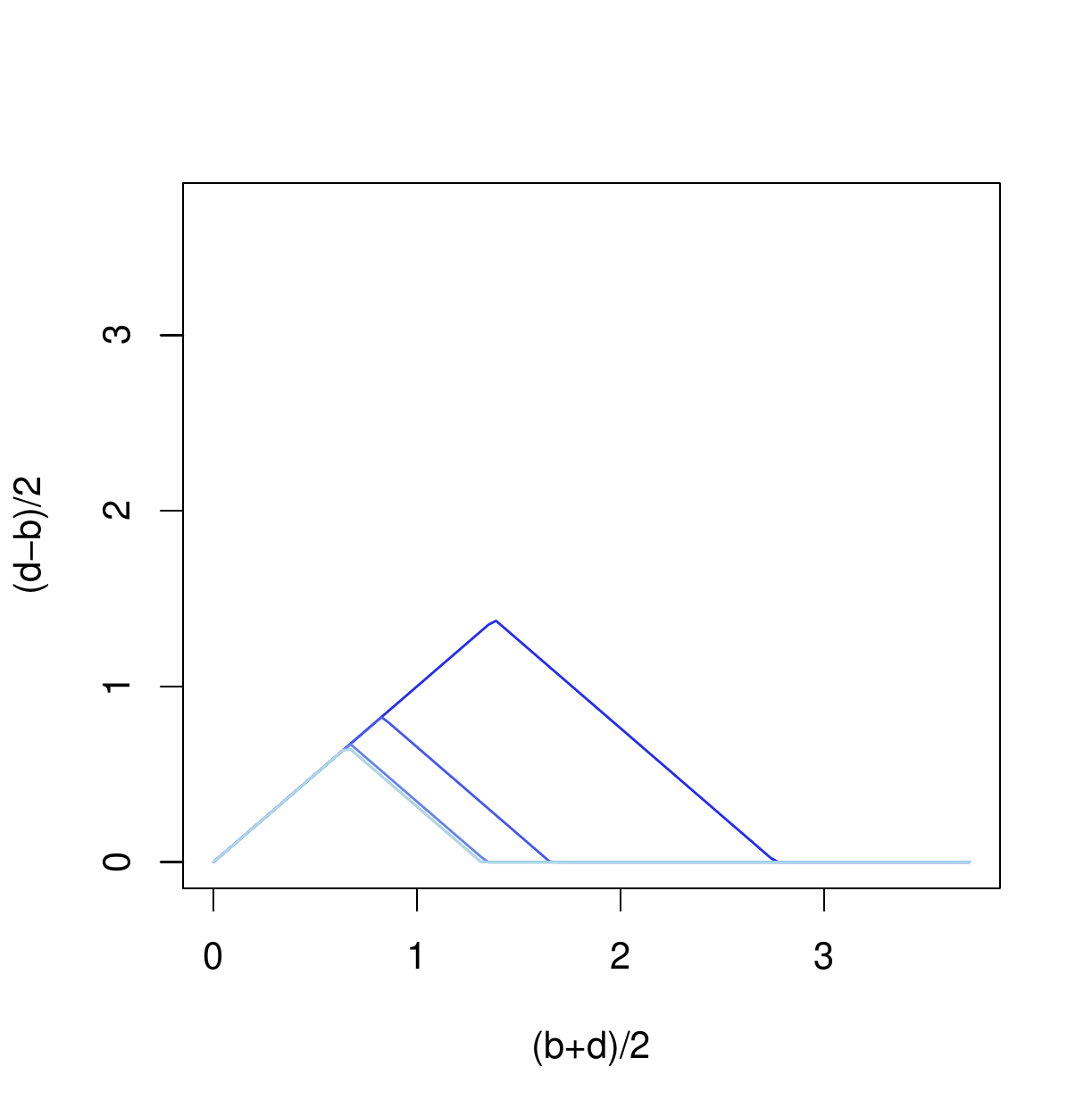}}
\subfloat[]{\includegraphics[width=7cm]{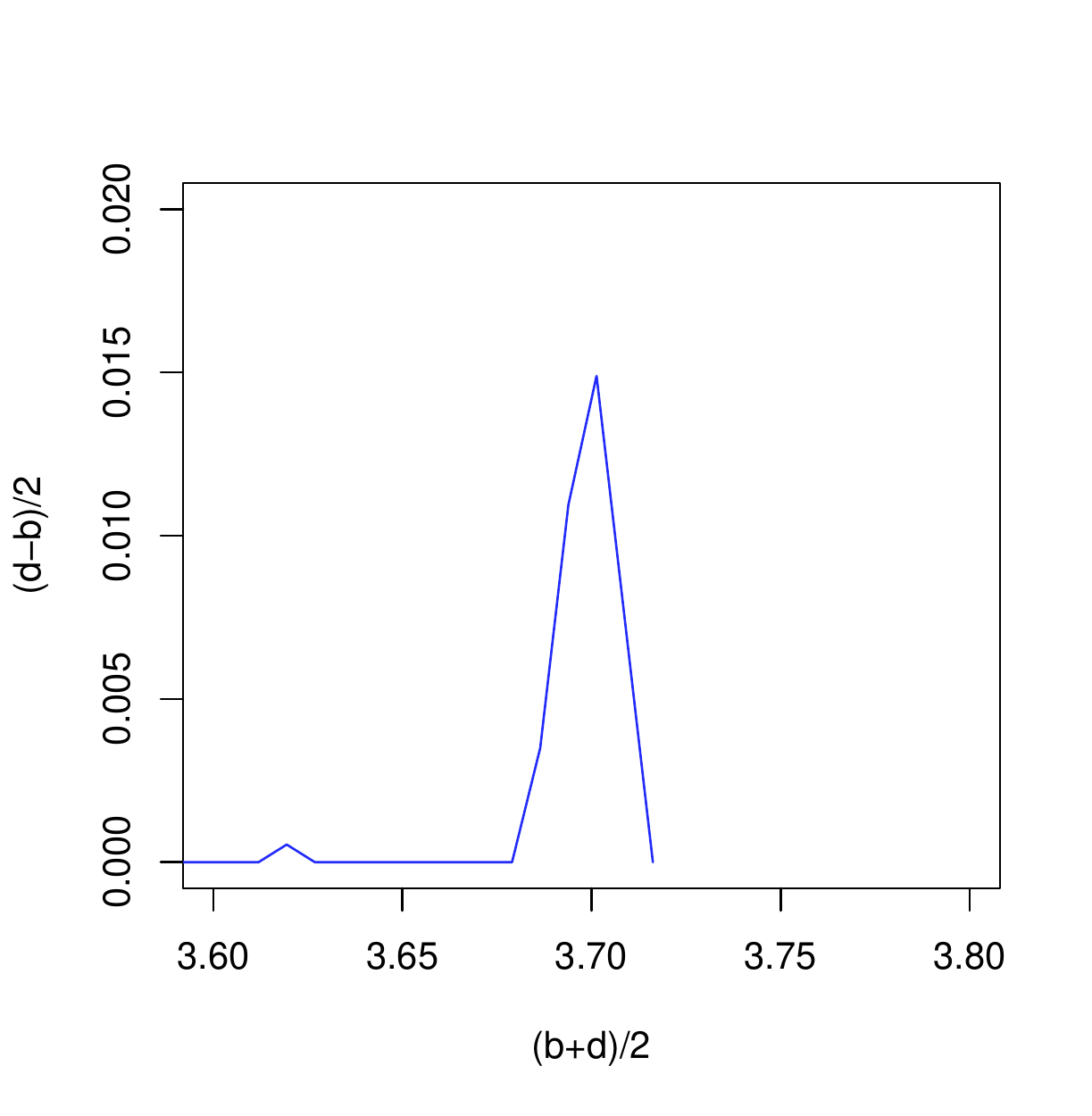}}
\caption{Rank function for connected components (a) and holes (b) Persistence Landscapes for connected components (c) and holes (d)}
\label{fig:perrank}
\end{figure}


More formally, a persistence landscape is a sequence of continuous, piecewise linear functions $\lambda(k,\cdot):\mathbb{R}^+\rightarrow\mathbb{R}^+, \,\,\,\, k=1,2,\dots$. Denote the set of `persistence points' in the persistence diagram by $D$. Then for each $p=(b,d)\in D$ define the triangular functions
\begin{equation*}
\Lambda_p(t)
=\Biggl\{
\begin{array}{rl}
t-b&t\in [b,\frac{b+d}{2}]\\
d-t&t\in (\frac{b+d}{2},d]\\
0&otherwise.\\
\end{array}
\label{eq:trianglefun}
\end{equation*}
Then, the persistence landscape of the persistence diagram is defined by
\begin{equation}
\lambda_D(k,t)=k\max_{p\in D}\Lambda_p(t), \,\,\,\,\,\,\,\, t\ge0, \,\, k\in \mathbb{N},
\end{equation}
where $k\max$ is the $k$th largest value in the set. 

Our test will be the contrast between the observed landscape and a landscape one would expect under the null hypothesis that the 3D structure is Poisson-Voronoi. For this \textit{mean landscape}, we use the conditional expectation of the landscape given that $N_{2D}=n_{2D}$ and approximate this using the simulation procedure described in Section \textbf{\ref{subsec:testper}}. To be more specific,
\begin{equation}
\bar \lambda_{D_j}(k,t)=\frac{1}{n}\sum_{i=1}^n \lambda_{D_j (i)}(k,t) \,\,\, j=0,1,\,\,\,\, t\ge0,
\end{equation}
where $n$ is the number of 2D Poisson-Voronoi sections generated with $N_{2D}=n_{2D}$.
Inspired by the approach proposed in \cite{robins2016}, the test statistics are then given by the distance between persistence landscapes and mean persistence landscapes using $L^2$ norm,
\begin{equation} \label{eq:testperland}
\begin{split}
 L_0&=\|\hat \lambda_{D_0}- \bar \lambda_{D_0}\|_2=\biggl[\sum_{k=1}^{n_{2D}-1}\int_{0}^T(\hat \lambda_{D_0}(k,t)- \bar \lambda_{D_0}(k,t))^2 \mathrm{dt}\biggr]^{\frac{1}{2}} \\
 L_1&=\|\hat \lambda_{D_1}- \bar \lambda_{D_1}\|_2=\biggl[\sum_{k=1}^{\infty}\int_{0}^T(\hat \lambda_{D_1}(k,t)- \bar \lambda_{D_1}(k,t))^2 \mathrm{dt} \biggr]^{\frac{1}{2}}.
\end{split}
\end{equation}
Here $\hat \lambda_{D_j}(k,\cdot), \,\,\,j=0,1$ is the $k$-th landscape for the connected components ($j=0$) and for the holes ($j=1$) for the 2D section under study.
If both $L_0$ and $L_1$ are less than the threshold quantiles, the Poisson-Voronoi hypothesis is not rejected.

\section{Bootstrap Confidence Interval for $\lambda$ and Quantiles of the model tests}
\label{sec:quantile}
In \cite{lorzhahn93}, the authors carry out a simulation for estimating the quantiles of the test statistics proposed there.
Cells of 3D spatial Poisson-Voronoi diagrams are generated with $\lambda=1$. Then, a random planar section of the 3D structure is taken and square observation windows are drawn in the section planes with an expected number of 50, 100, 150 and 200 cells, respectively.

We provide an expression for the distribution of any test statistic given the number of observed cells in the section, separating a part that depends on the parameter $\lambda$ and a part that does not.
We consider the situation where we see a window (with known shape and size) of a $2D$ planar section of a $3D$ Poisson-Voronoi diagram in a $3D$ object of known geometry.
As before, denote by $N_{3D}$ the number of cells in the $3D$ object and $N_{2D}$ the number of $2D$ cells visible in the $2D$ window. Lemma \textbf{\ref{lemma:lambdadep}} below  gives an expression of the null distribution of a test statistic $T$, given $n_{2D}$ cells are observed in the section. It separates a part that depends on the intensity parameter $\lambda$ and a part that does not.

\begin{lemma}
Let $T$ denote a general model test for the Poisson-Voronoi assumption validation. The conditional probability $P_\lambda(T\ge t|N_{2D}=n_{2D})$ can be expressed as
\begin{equation} \label{eq1lambda}
\small
\begin{split}
P_\lambda(T\ge t|N_{2D}=n_{2D}) & =\frac{\sum_{k=n_{2D}}^\infty P(T\ge t |N_{3D}=k,\, N_{2D}=n_{2D})\, P(N_{2D}=n_{2D}|N_{3D}=k)\frac{(\lambda \mathcal{V})^k}{k!}}{\sum_{j=n_{2D}}^\infty P(N_{2D}=n_{2D}|N_{3D}=j)\frac{(\lambda \mathcal{V})^j}{j!}}
\end{split}
\end{equation}
\label{lemma:lambdadep}
\end{lemma}

\begin{proof}
\begin{equation} \label{eq2lambda}
\small
\begin{split}
P_\lambda(T\ge t|N_{2D}=n_{2D}) & = \sum_{k=0}^\infty P_\lambda(T\ge t,\, N_{3D}=k| N_{2D}=n_{2D}) \\
& = \sum_{k=n_{2D}}^\infty P_\lambda(T\ge t,\, N_{3D}=k| N_{2D}=n_{2D}) \\
 & =  \sum_{k=n_{2D}}^\infty P_\lambda(T\ge t |N_{3D}=k,\, N_{2D}=n_{2D})\, P_\lambda(N_{3D}=k|N_{2D}=n_{2D})\\
  & =  \sum_{k=n_{2D}}^\infty P(T\ge t |N_{3D}=k,\, N_{2D}=n_{2D})\, P_\lambda(N_{3D}=k|N_{2D}=n_{2D}).
\end{split}
\end{equation}
In the last equality the $\lambda$-dependence disappears from the first factor, because, conditionally on $N_{3D}$, the distribution of $T$ does not depend on $\lambda$.
The $\lambda$-dependent part in eq. \textbf{\ref{eq2lambda}} can be made more explicit also using that conditionally on $N_{3D}$, the distribution of $N_{2D}$ does not depend on $\lambda$:
\begin{equation} \label{eq3lambda}
\small
\begin{split}
 P_\lambda(N_{3D}=k|N_{2D}=n_{2D})&=\frac{P(N_{2D}=n_{2D}|N_{3D}=k)P_\lambda(N_{3D}=k)}{P_\lambda(N_{2D}=n_{2D})}\\
& = \frac{P(N_{2D}=n_{2D}|N_{3D}=k)P_\lambda(N_{3D}=k)}{\sum_{j=n_{2D}}^\infty P(N_{2D}=n_{2D}|N_{3D}=j)P_\lambda(N_{3D}=j)}\\
& = \frac{P(N_{2D}=n_{2D}|N_{3D}=k)\frac{(\lambda \mathcal{V})^k}{k!}}{\sum_{j=n_{2D}}^\infty P(N_{2D}=n_{2D}|N_{3D}=j)\frac{(\lambda \mathcal{V})^j}{j!}}.
\end{split}
\end{equation}

Combining eqs. \textbf{\ref{eq2lambda}} and \textbf{\ref{eq3lambda}} yields eq. \textbf{\ref{eq1lambda}}.
\end{proof}

For computing p-values in practice, the value of $\lambda$ is needed. 
In order to take into account the uncertainty in the estimate of $\lambda$ while computing p-values for model tests, we compute a 90\% confidence interval for $\lambda$.
To do this, as this value is not known, we propose a bootstrap approach.  
More precisely, we want to compute 
\begin{equation} \label{eqboot}
\small
\begin{split}
 P_\lambda(\sqrt{\hat{\lambda}}-\sqrt{\lambda}\le u)&=\sum_{k=0}^\infty P_\lambda(\sqrt{\hat{\lambda}}-\sqrt{\lambda}\le u, N_{3D}=k)\\
& = \sum_{k=0}^\infty P_\lambda(\sqrt{\hat{\lambda}}-\sqrt{\lambda}\le u|N_{3D}=k) P_\lambda(N_{3D}=k)
\end{split}
\end{equation}

The procedure can be summarized as follows: first we estimate $\lambda$ from a real 2D image, using $\hat \lambda_a$ (eq. \textbf{\ref{eq:estimators}}). 
For computing $P_\lambda(N_{3D}=k)$, $P_{\hat \lambda_a}(N_{3D}=k)$ is then used.
Secondly, for computing $P_\lambda(\sqrt{\hat{\lambda}}-\sqrt{\lambda}\le u|N_{3D}=k)$,  $10\,000$ Poisson-Voronoi diagrams for each realization of a Poisson process with $\hat \lambda_a$ in a cube are generated. Then a 2D section from each 3D diagram is randomly taken and the number of cells in the section is used for estimating $\lambda$. Next, the probability of having exactly $k$ cells in 3D, $P(N_{3D}=k)$, is used as weight for computing a weighted mean cumulative distribution function. Finally, a square root transformation for normalizing and stabilizing the variance is used for computing the confidence set \cite{sahai1993}:

\begin{equation}
P[\sqrt{\hat{\lambda}}-l_{0.95}\le \sqrt{\lambda} \le \sqrt{\hat{\lambda}}-l_{0.05}]\approx 0.90
\end{equation}
As an example, if $\hat \lambda_a=0.2$ (as in the application shown in Section \textbf{\ref{sec:application}} $n_{2D}=50$ and window size$=10\times10$), the resulting $90\%$-confidence set is given by:

\begin{equation}
[0.1498;0.2439]
\label{eq:confintlambda}
\end{equation}

Having a confidence set for $\lambda$ at hand, the next step is to compute the null distribution described in \textbf{Lemma \ref{lemma:lambdadep}} for the various test statistics. These p-values depend on $\lambda$, but we can consider these for all $\lambda$ in the confidence set constructed. We start with the coefficient of variation as test statistic (eq.\textbf{ \ref{eq:cv}}). 
In Figure \textbf{\ref{fig:cdfcvcomp}} it is possible to see the difference between the cumulative distribution functions of the coefficient of variation of the 2D sectional cells area unconditioned and conditioned on seeing exactly $50$ cells in the 2D section. Moreover, the green dotted lines represent the cumulative distribution function of the coefficient of variation for the lower and upper bounds of the $\lambda$ confidence set. Note that the distance between the two cdfs is small, showing that the approach of \cite{lorzhahn93} to use an unconditional distribution in this particular setting leads to comparable results. In table \textbf{\ref{tab:quantilecv}}, quantiles for the conditional distribution of the CV of the cells area are shown ($\hat \lambda=0.2$).

\begin{figure}[!h]
\centering
\includegraphics[width=8cm]{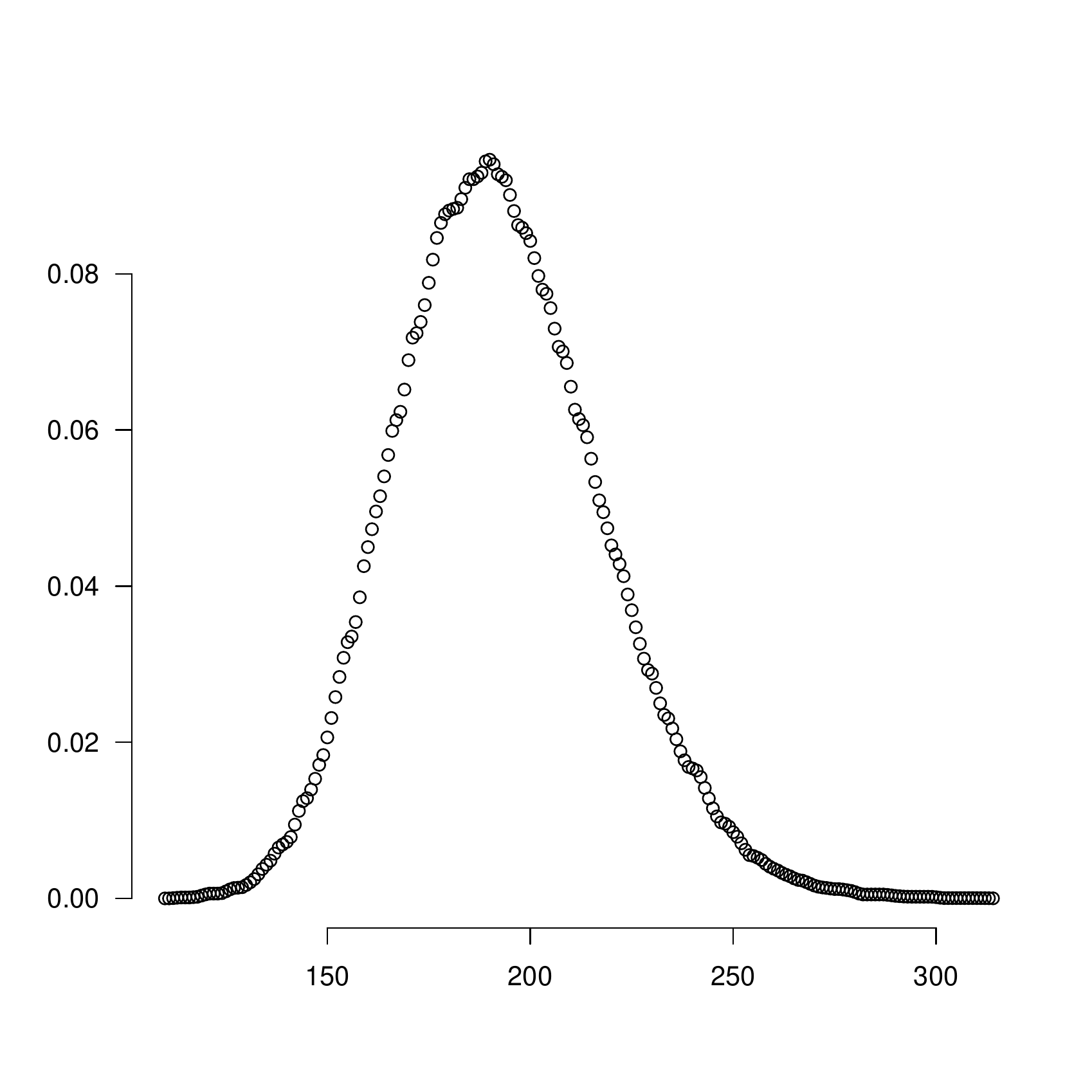}
\caption{Monte Carlo approximation of $P(N_{2D}=50|N_{3D}=k)$ }
\label{fig:p50n3d}
\end{figure}

\begin{figure}[!h]
\centering
\includegraphics[width=10cm]{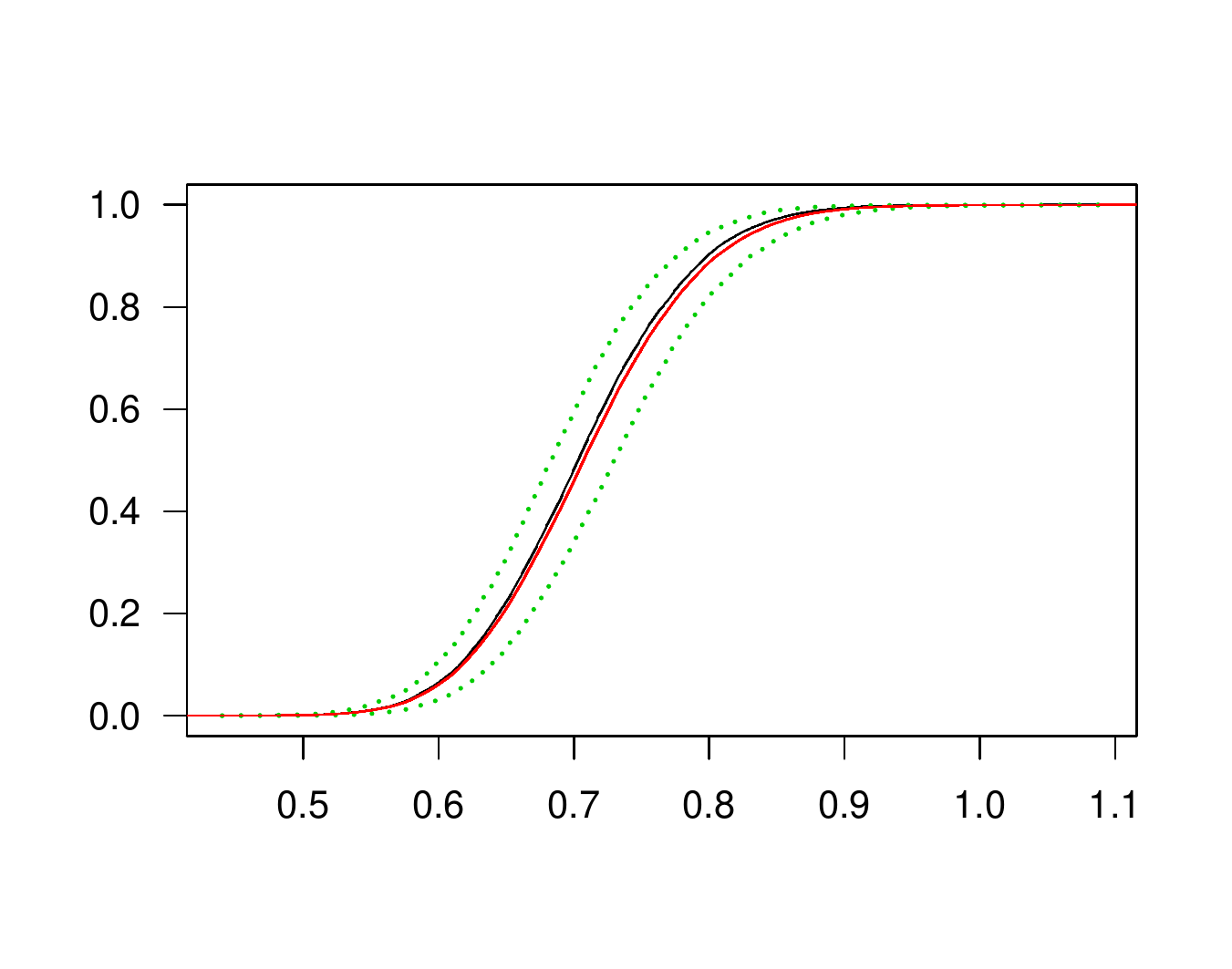}
\caption{Cumulative distribution function of the coefficient of variation of the 2D sectional cells area conditioned on $N_{2D}=50$ (black line; green dotted lines are obtained using the upper and lower limit of the confidence set for $\lambda$ (eq.\textbf{\ref{eq:confintlambda}})) and unconditioned (red line)}
\label{fig:cdfcvcomp}
\end{figure}

\begin{table}[!h]
\centering
\caption{Quantiles of the conditional distribution of the coefficient of variation of the 2D sectional cells area given that $N_{2D}=50$, ($\lambda=0.2$)}
\small
\begin{tabular}{rrrrrrrrrrrrr}
  \hline
$\alpha$& 0.005& 0.01& 0.0125&0.025&0.05& 0.1&0.9&0.95&0.975&0.9875&0.99&0.995 \\
  \hline
$c_\alpha$&0.531& 0.547& 0.553 &0.571& 0.591 &0.615 &0.798& 0.826& 0.853 &0.875& 0.883& 0.903\\
\end{tabular}
\label{tab:quantilecv}
\end{table}


In Figure \textbf{\ref{fig:cdf50c}}, the conditional weighted mean CDF for cells area (black line), its confidence bands (green dotted lines) and the unconditional mean are shown. More precisely we define
\begin{equation}
\begin{split}
\bar F_{\lambda \,n_{2D}}(x)&=\mathbb{E}_\lambda\{F_{N_{2D}}(x)|N_{2D}=n_{2D}\}=\mathbb{E}_\lambda\{\mathbb{E}(F_{N_{2D}}(x)|N_{2D}=n_{2D},\,N_{3D})\}\\
&=\sum_{k=n_{2D}}^\infty P_\lambda(N_{3D}=k|N_{2D}=n_{2D})\cdot\mathbb{E}(F_{N_{2D}=n_{2D},N_{3D}=k}(x)).
\end{split}
\end{equation}
Where $F_{N_{2D}=n_{2D},N_{3D}=k}(x)$ is the empirical distribution function of the areas given $k$ cells in $3D$ structure and $n_{2D}$ visible on the slice.
The same type of expression is used also for $\bar \lambda_{D_0}(1,t)$ and $\bar \lambda_{D_1}(1,t)$.

In Figure \textbf{\ref{fig:cdfecdfcomp}} the CDF of the test based on the supremum distance between ecdfs of the 2D sectional cells area is shown.
As for the test based on coefficient of variation, the difference between conditional and unconditional approach is relatively small.

Switching to the test based on persistence landscapes, Fig. \textbf{\ref{mlan0cond}-\ref{mland1cond}} are visualizations of $k$ mean persistence landscapes conditioned on $N_{2D}=50$ for connected components and holes respectively, when $\hat \lambda=0.2$.
Instead, Fig. \textbf{\ref{maxmeanland0cb2}-\ref{maxmeanland1cb2}} are the conditional maximum weighted means (black lines) and their confidence bands (green dotted lines). 

In Figures \textbf{\ref{fig:cdfl0comp}-\ref{fig:cdfl1comp}}  the CDF of the test based on the $L_2$ distance between persistence landscapes (connected components and holes) are shown.
Also in this case, the difference between conditional and unconditional approach seems to be irrilevant.

\begin{figure}[!h]
\centering
\includegraphics[width=8cm]{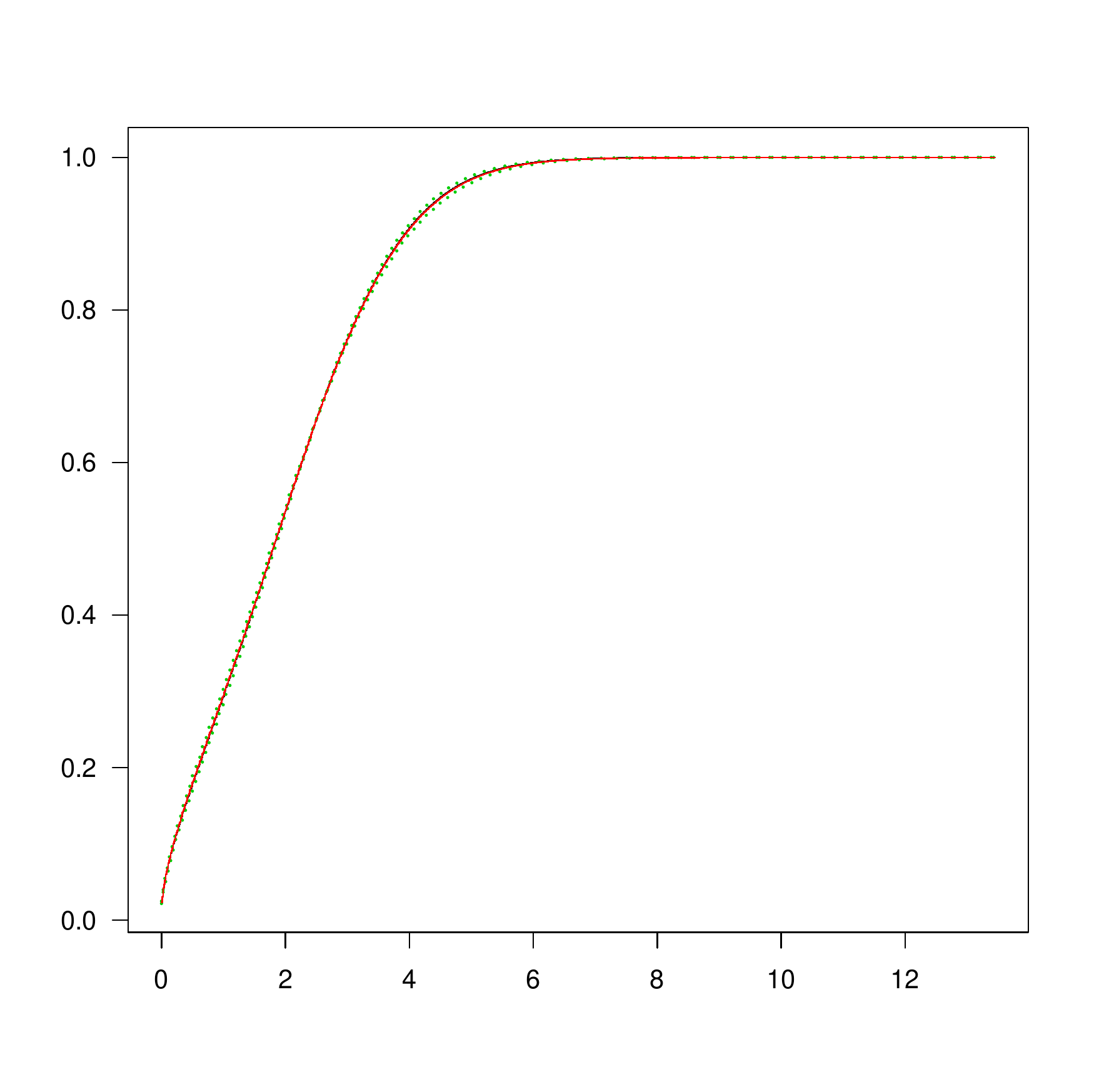}
\caption{Cumulative distribution function of the 2D sectional cells area conditioned on $N_{2D}=50$ (black line; green dotted lines are obtained using the upper and lower limit of the confidence set for $\lambda$ (eq.\textbf{\ref{eq:confintlambda}})) and unconditioned (red line)}
\label{fig:cdf50c}
\end{figure}

\begin{figure}[!h]
\centering
\includegraphics[width=10cm]{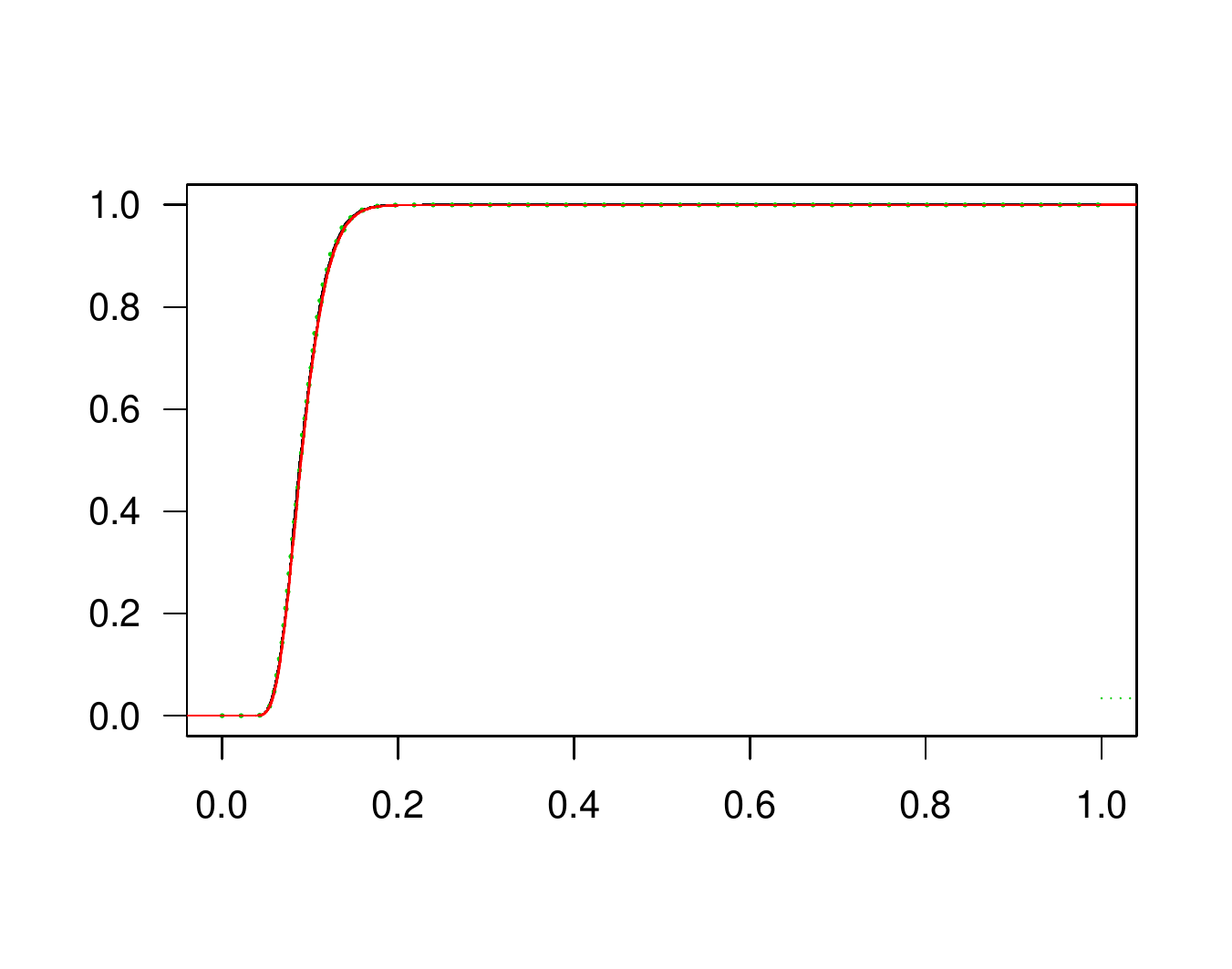}
\caption{Cumulative distribution function of the ecdf test of the 2D sectional cells area conditioned on $N_{2D}=50$ (black line; green dotted lines are obtained using the upper and lower limit of the confidence set for $\lambda$ (eq.\textbf{\ref{eq:confintlambda}})) and unconditioned (red line)}
\label{fig:cdfecdfcomp}
\end{figure}

\begin{figure}[!h]
\centering
\includegraphics[width=9cm]{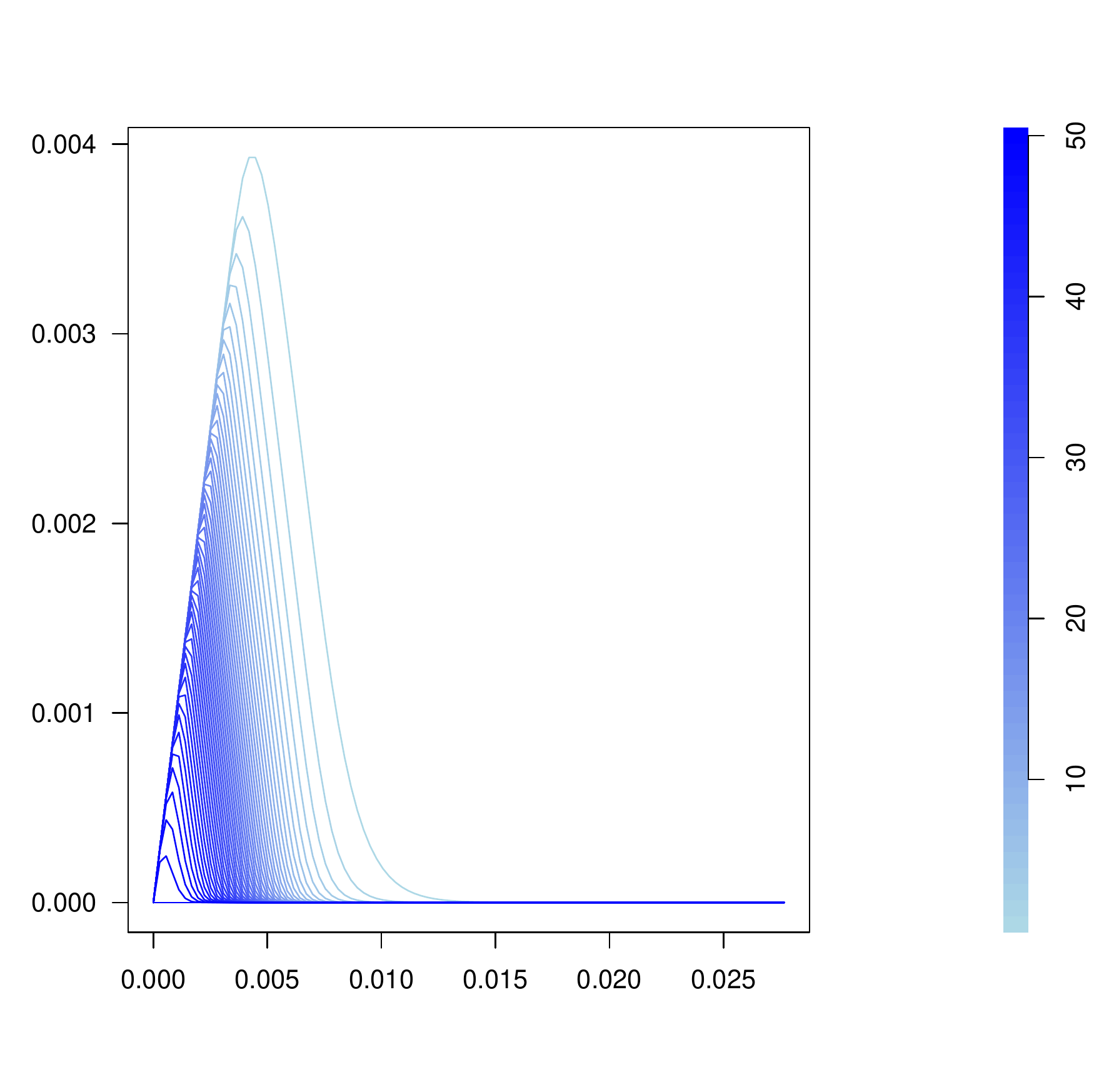}
\caption{$k$ mean landscapes conditioned on $N_{2D}=50$ (connected components) ($\hat \lambda=0.2$)}
\label{mlan0cond}
\end{figure}

\begin{figure}[!h]
\centering
\includegraphics[width=9cm]{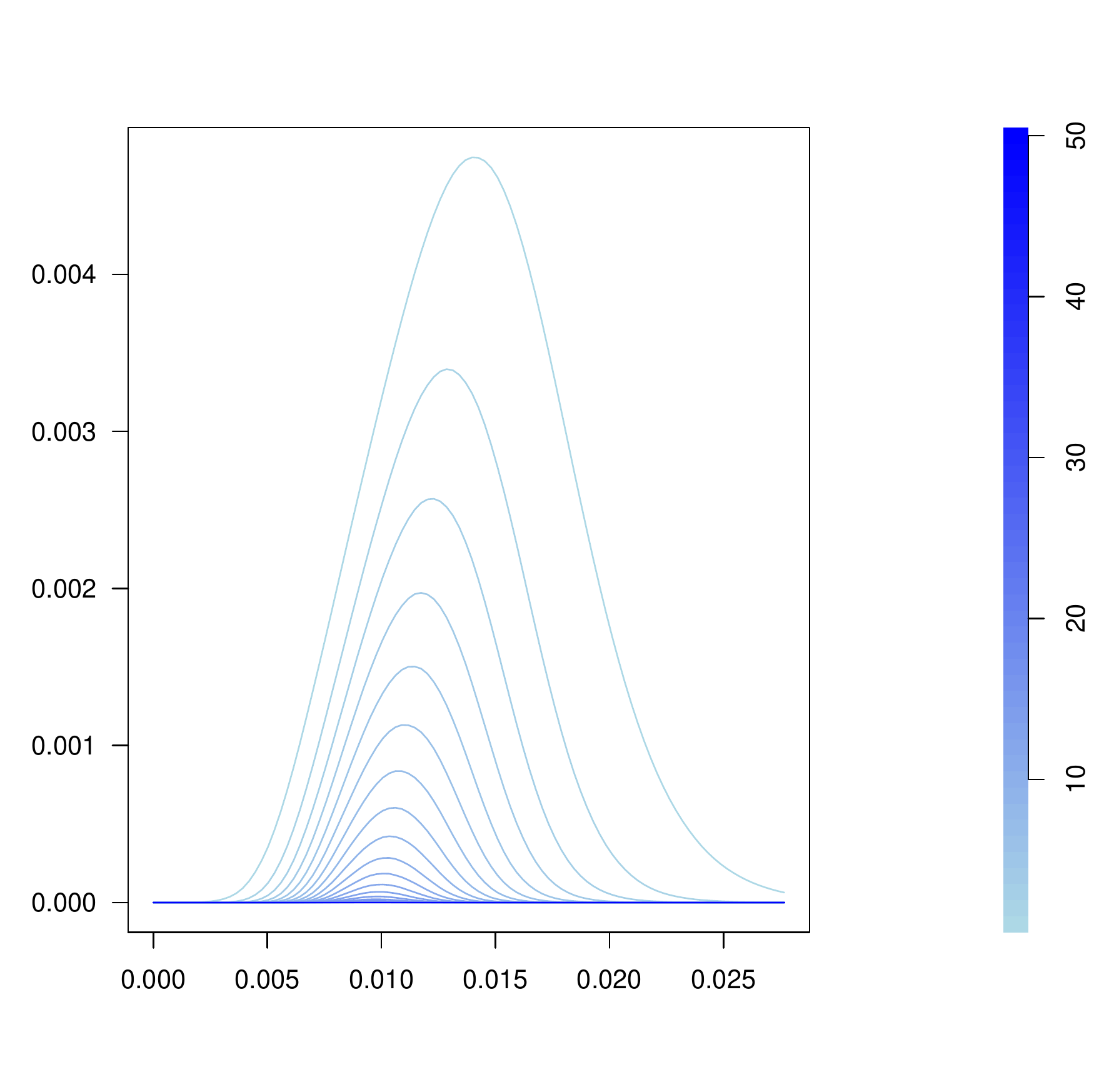}
\caption{$k$ mean landscapes conditioned on $N_{2D}=50$  (holes) ($\hat \lambda=0.2$)}
\label{mland1cond}
\end{figure}

\begin{figure}[!h]
\centering
\includegraphics[width=8cm]{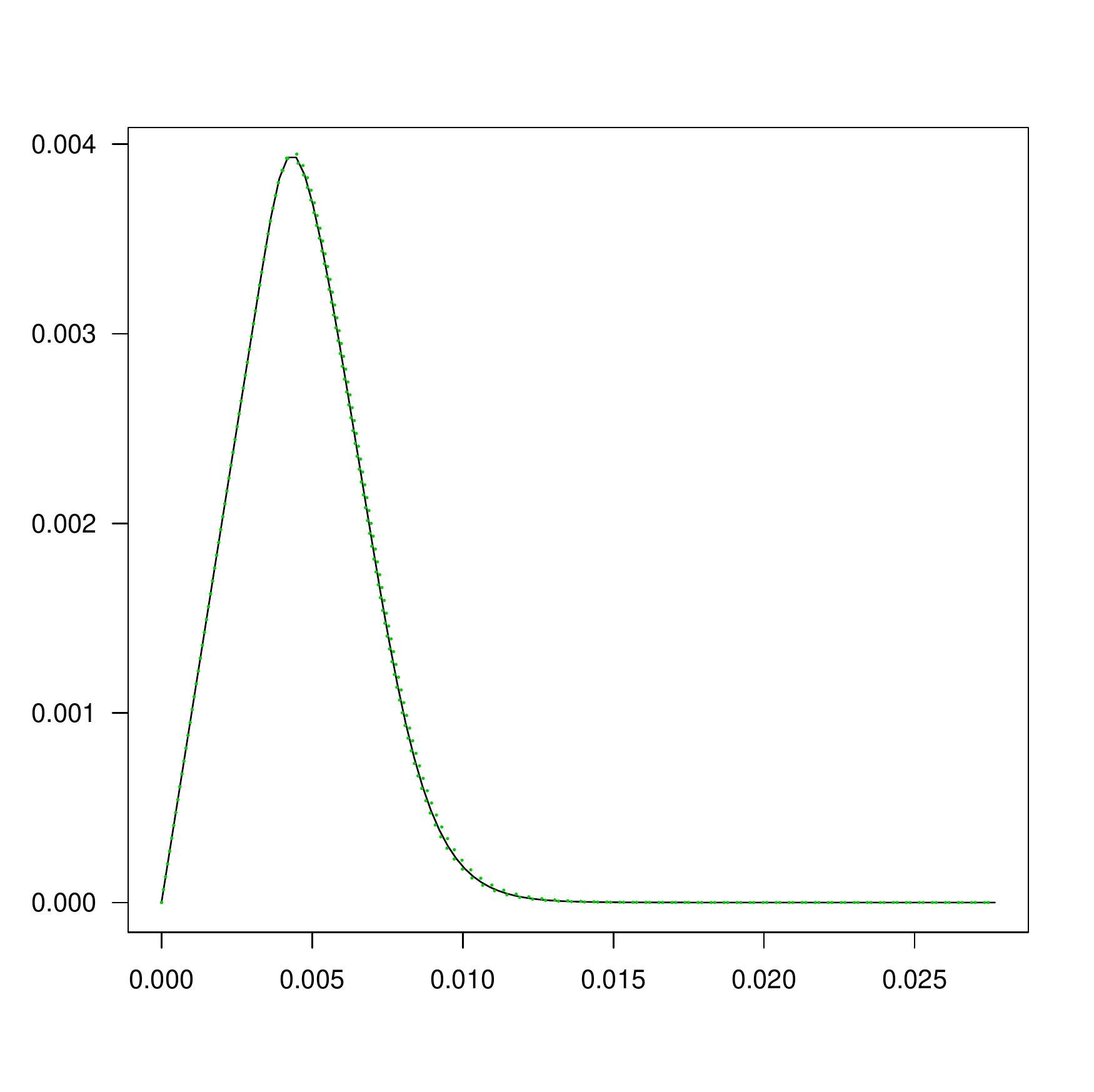}
\caption{Max weighted mean landscape (connected components) for sections with exactly $50$ 2D sectional cells (black line; green dotted lines are obtained using the upper and lower limit of the confidence set for $\lambda$ (eq.\textbf{\ref{eq:confintlambda}}))}
\label{maxmeanland0cb2}
\end{figure}

\begin{figure}[!h]
\centering
\includegraphics[width=8cm]{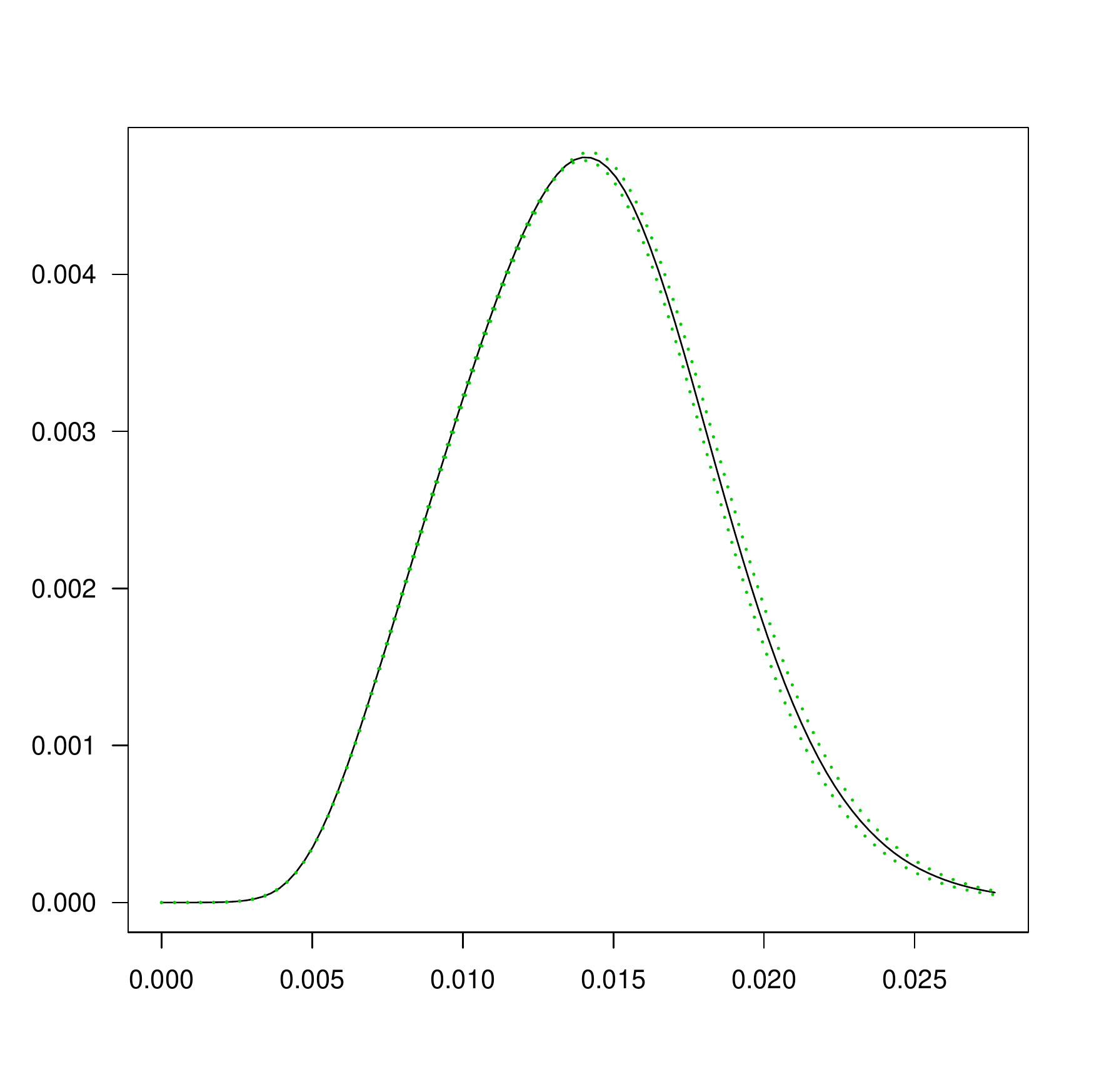}
\caption{Max weighted mean landscape (holes) for sections with exactly $50$ 2D sectional cells (black line; green dotted lines are obtained using the upper and lower limit of the confidence set for $\lambda$ (eq.\textbf{\ref{eq:confintlambda}}))}
\label{maxmeanland1cb2}
\end{figure}

\begin{figure}[!h]
\centering
\includegraphics[width=8cm]{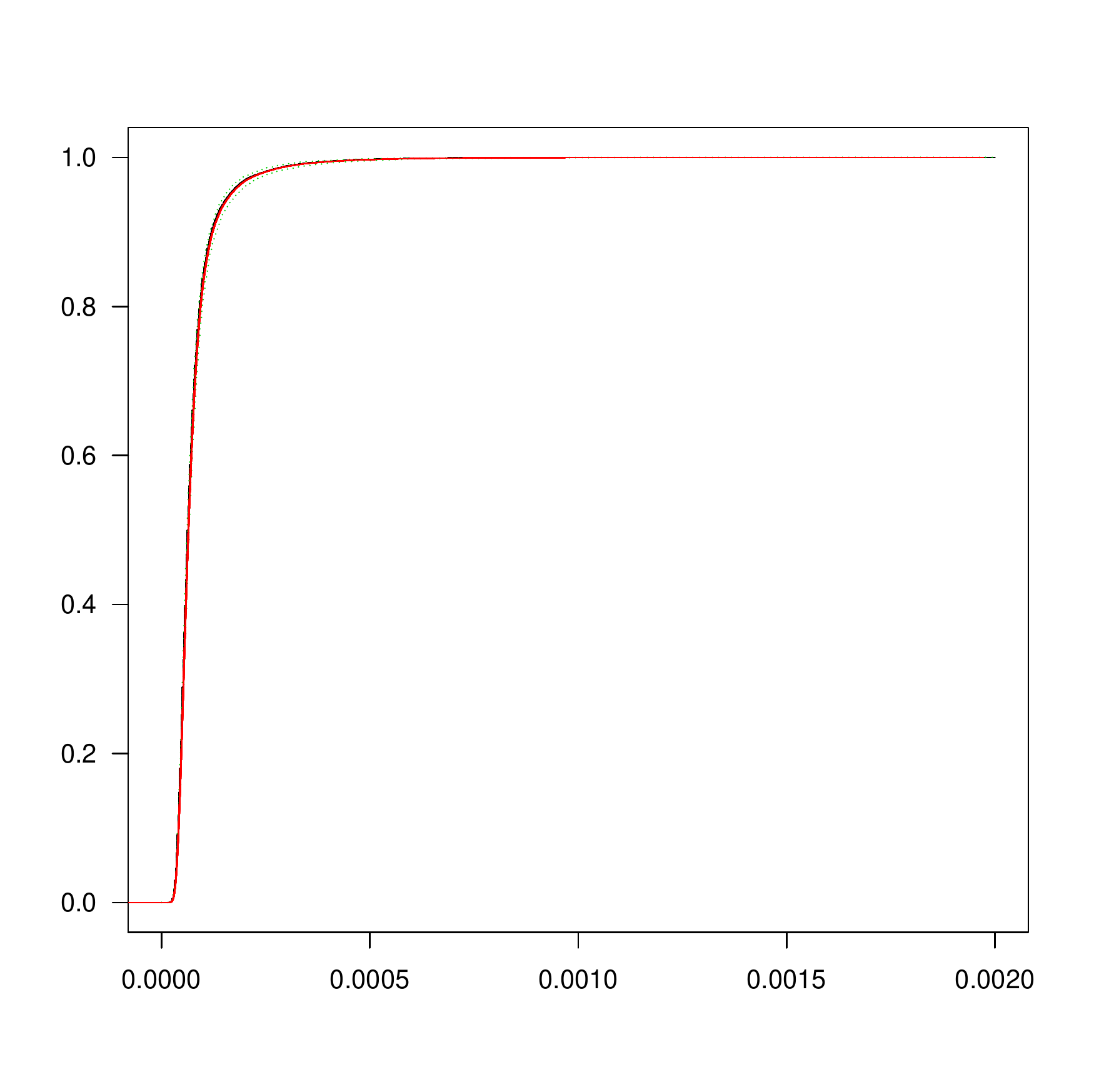}
\caption{Cumulative distribution function of the test statistic based on the $L_2$ distance between persistence landscapes $L_0$, (\textbf{\ref{eq:testperland}}), of the 2D sectional cells area conditioned on $N_{2D}=50$ (black line; green dotted lines are obtained using the upper and lower limit of the confidence set for $\lambda$ (eq.\textbf{\ref{eq:confintlambda}})) and unconditioned (red line)}
\label{fig:cdfl0comp}
\end{figure}

\begin{figure}[!h]
\centering
\includegraphics[width=8cm]{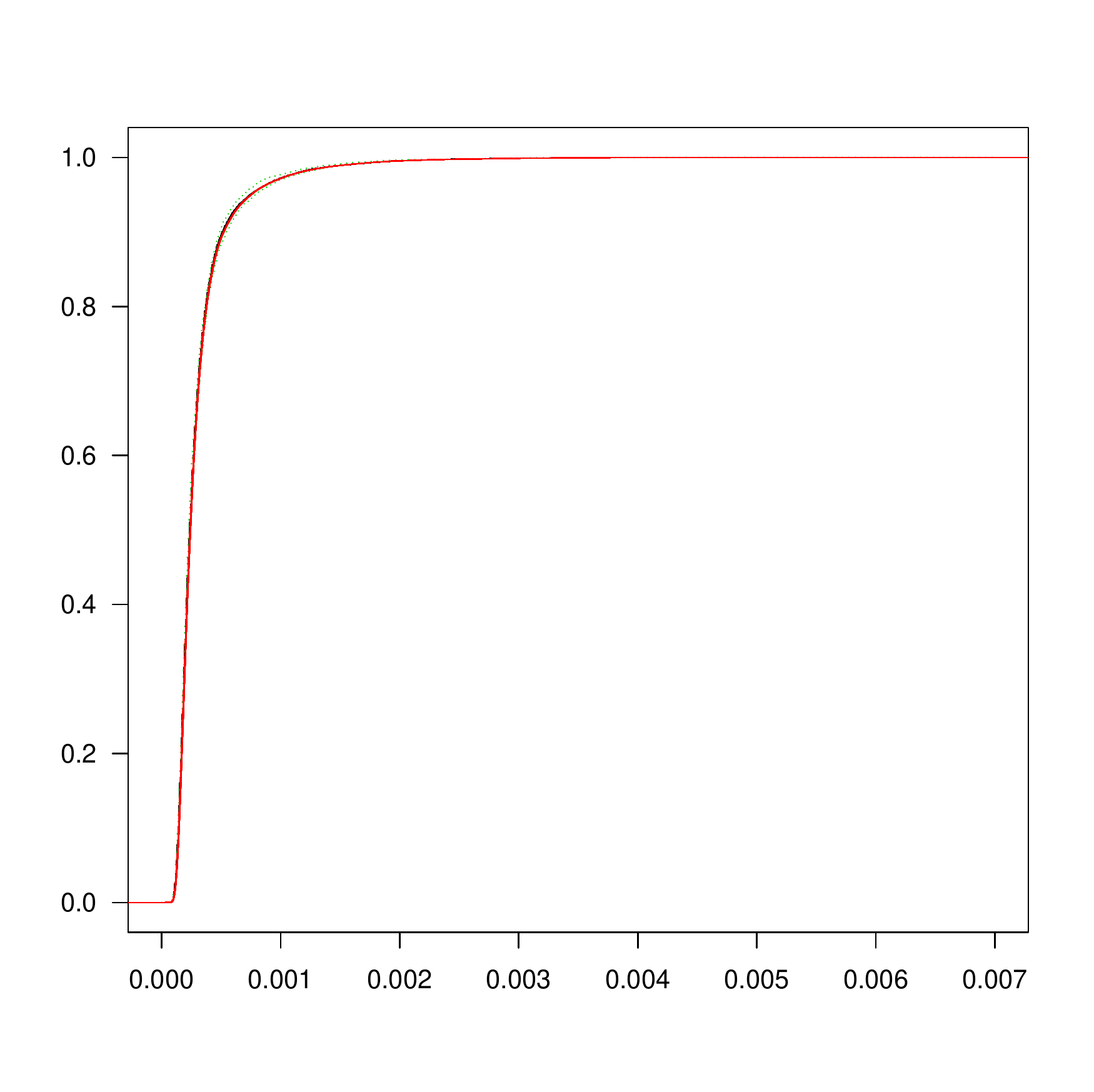}
\caption{Cumulative distribution function of the test statistic based on the $L_2$ distance between persistence landscapes $L_1$, (\textbf{\ref{eq:testperland}}), of the 2D sectional cells area conditioned on $N_{2D}=50$ (black line; green dotted lines are obtained using the upper and lower limit of the confidence set for $\lambda$(eq.\textbf{\ref{eq:confintlambda}})) and unconditioned (red line)}
\label{fig:cdfl1comp}
\end{figure}

For computing the quantiles of the distribution of the model tests based on CDF and on persistence landscape, we use a `leave one out' procedure. Here we use the $B$ 2D slices generated as follows:
\begin{itemize}
\item For the test based on CDFs difference (\textbf{\ref{eq:testcdf}}):
\begin{equation}\label{eq:cdfquant}
d_i=\sup_{x\in \mathbb{R}}|\bar F_{\lambda \,n_{2D}(-i)}(x)-\hat F_{n_{2D}(i)}(x)|, \,\,\,\,\,\, 1\le i\le B
\end{equation}
\item For the test based on persistence landscapes difference (\textbf{\ref{eq:testperland}}):
\begin{equation} \label{eq:perlandquant}
\begin{split}
 l_{0(i)}&=\biggl[\sum_{k=1}^{n_{2D}-1}\int_{0}^T(\hat \lambda_{D_0(i)}(k,t)- \bar \lambda_{D_0(-i)}(k,t))^2 \mathrm{dt}\biggr]^{\frac{1}{2}} \,\,\,\,\,\, 1\le i\le B\\
 l_{1(i)}&=\biggl[\sum_{k=1}^\infty\int_{0}^T(\hat \lambda_{D_1(i)}(k,t)- \bar \lambda_{D_1(-i)}(k,t))^2 \mathrm{dt}\biggr]^{\frac{1}{2}}  \,\,\,\,\,\, 1\le i\le B
\end{split}
\end{equation}
\end{itemize}
Here $\hat F_{n_{2D}(i)}$, $\hat \lambda_{D_0(i)}$ and $\hat \lambda_{D_1(i)}$ are the empirical results for the section $i$ and $\bar F_{n_{2D}(-i)}$, $\bar \lambda_{D_0(-i)}$ and $\bar \lambda_{D_1(-i)}$ are the mean result computed for all the $B$ sections leaving out the $i$-th.

\begin{table}[!h]
\centering
\caption{Quantiles of the conditional distribution of the test based on the difference between cumulative distribution functions of the 2D sectional cells area given that $N_{2D}=50$, ($\hat \lambda=0.2$)}
\footnotesize
\begin{tabular}{rrrrrrrrrrrrr}
  \hline
$\alpha$& 0.005& 0.01& 0.0125&0.025&0.05& 0.1&0.9&0.95&0.975&0.9875&0.99&0.995 \\
  \hline
$d_\alpha$&0.047& 0.050& 0.051&0.054&0.058& 0.064 &0.123 &0.135& 0.146&0.155&0.159& 0.168\\
\end{tabular}
\label{tab:quantilet0}
\end{table}

\begin{table}[!h]
\centering
\caption{Quantiles of the conditional distribution of the test based on the difference between the observed landscapes and the conditional mean landscapes (connected components) of the 2D sectional cells area given that $N_{2D}=50$, ($\hat \lambda=0.2$)}
\footnotesize
\begin{tabular}{lrrrrrrrrrrrr}
  \hline
$\alpha$& 0.005& 0.01& 0.0125&0.025&0.05& 0.1&0.9&0.95&0.975&0.9875&0.99&0.995 \\
  \hline
$l_{0\alpha}\times 10^{-5}$&$2.402$&$ 2.602$& $2.802 $&$3.003$&$3.403$& $3.803 $&10 &20& 20&30&30&40\\
\end{tabular}
\label{tab:quantilet0}
\end{table}

\begin{table}[!h]
\centering
\caption{Quantiles of the conditional distribution of the test based on the difference between the observed landscapes and the conditional mean landscapes (holes) of the 2D sectional cells area given that $N_{2D}=50$, ($\hat \lambda=0.2$)}
\footnotesize
\begin{tabular}{rrrrrrrrrrrrr}
  \hline
$\alpha$& 0.005& 0.01& 0.0125&0.025&0.05& 0.1&0.9&0.95&0.975&0.9875&0.99&0.995 \\
  \hline
$l_{1\alpha}\times 10^{-5}$&$9.109$&$ 9.810$& $9.810$ &$10$&$10$& $10 $&50 &70& 100&140&150& 190\\
\end{tabular}
\label{tab:quantilet1}
\end{table}

\clearpage
\section{Application to single-phase alumina ceramics}
In \cite{lorzhahn93}, it is stated that single-phase microstructures, e.g. alumina ceramics, can be well approximated by Poisson-Voronoi diagrams.
Using the same images shown in \cite{lorzhahn93}, the tests proposed in the previous section (Sec.\textbf{ \ref{sec:test}}) are performed.

First all the cells in the images (Fig. \textbf{\ref{fig:lorzhahn93}} (a)) are involved in tests computation. Hereafter, for illustrative purposes and for a better comparison with the theoretical results shown in the previous section, we decide to consider just part of the images used in \cite{lorzhahn93}. In fact, the original window size is reduced until exactly $50$ cells are visible or partially visible (Fig.\textbf{ \ref{fig:lorzhahn93}} (b)). In Tables\textbf{ \ref{tab:lorztestall}-\ref{tab:lorztest}}, the test statistics and the p values (values in brackets) are shown for the four model tests and following the two different approaches.
Figures (\textbf{\ref{fig:cdf50ctestall}-\ref{fig:landtest1all}}) and (\textbf{\ref{fig:cdf50ctest}-\ref{fig:landtest1}}) are graphical representations of the cumulative distribution function test and the persistence approach steps.
In particular, for applying the test based on the difference between persistence landscapes, we take the center of mass of the cells in the images (Fig. \textbf{\ref{fig:lorzpp50all}, \ref{fig:lorzpp50}}), then compute the persistence diagrams (Fig. \textbf{\ref{fig:perdiaglorzall},\ref{fig:perdiaglorz}}) and finally the persistence landscapes (Fig. \textbf{\ref{fig:landtest0all}-\ref{fig:landtest1all}, \ref{fig:landtest0}-\ref{fig:landtest1}}) as explained in Section \textbf{\ref{subsubsec:perapproach}}.

Results using the two different approaches lead to slightly different results regarding the first two images (Fig. \textbf{\ref{fig:lorzhahn93}} 1(a), 1(b), 2(a), 2(b)).
For the first image, considering all the cells, the coefficient of variation test and the test based on the cdf of cells area suggest that the Poisson-Voronoi model could be reasonably used for approximating alumina ceramics; instead, looking at the cuts, the hypothesis is rejected by both tests.
For the second image, the coefficient of variation test based on all the cells is in agreement with the results obtained for the reduced sections; just the test based on the cdf considering all the cells does not reject the Poisson-Voronoi hypothesis.
Using tests from persistence approach instead, the use of Poisson-Voronoi model is discouraged in both cases.

\label{sec:application}
\begin{figure}[!h]
\centering
\captionsetup[subfigure]{labelformat=empty}
\subfloat[1a]{\includegraphics[width=5cm]{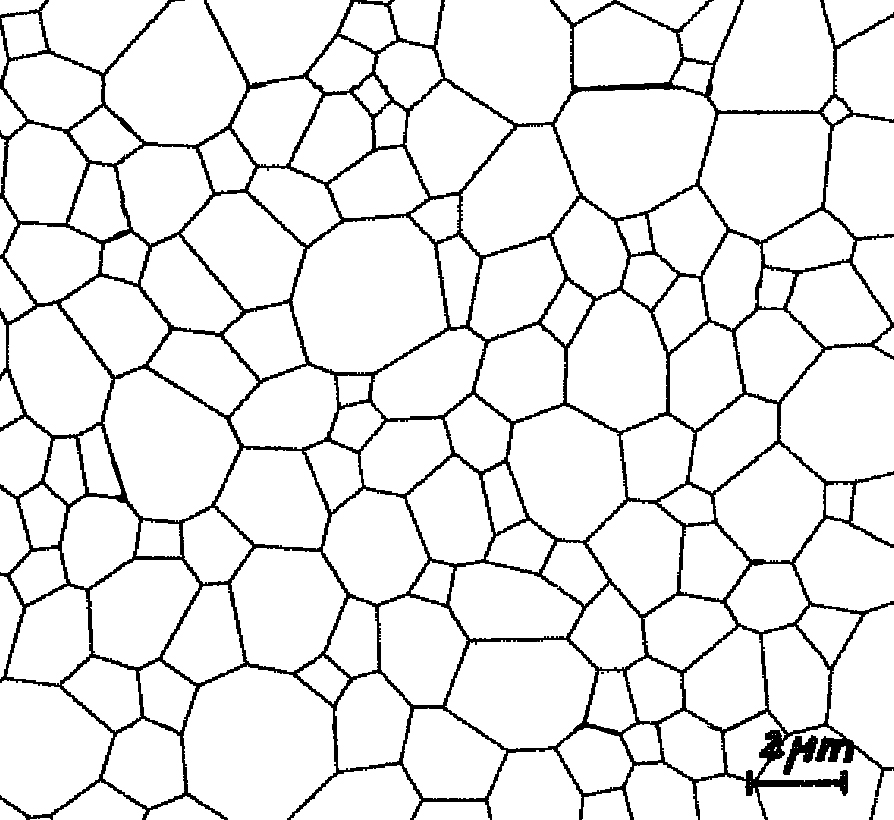}}
\subfloat[1b]{\includegraphics[width=4cm]{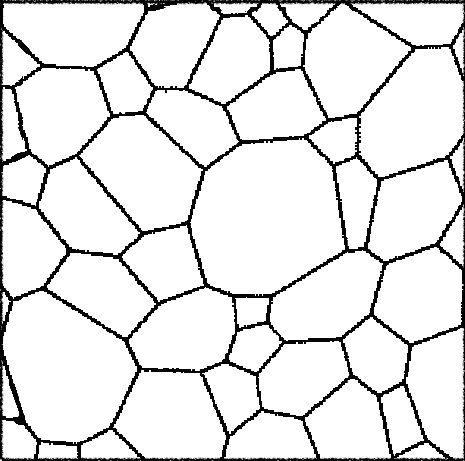}}

\subfloat[2a]{\includegraphics[width=5cm]{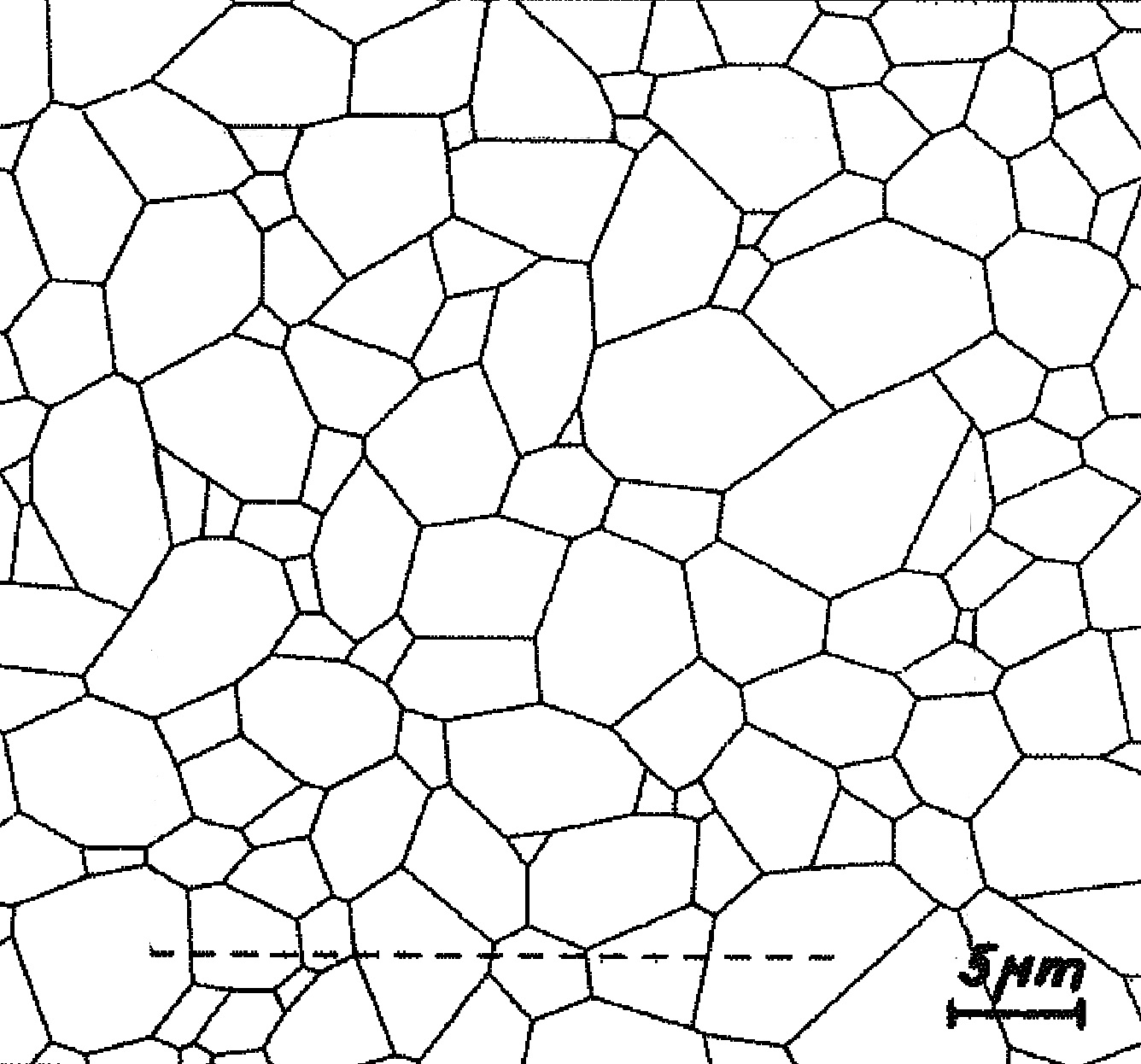}}
\subfloat[2b]{\includegraphics[width=4cm]{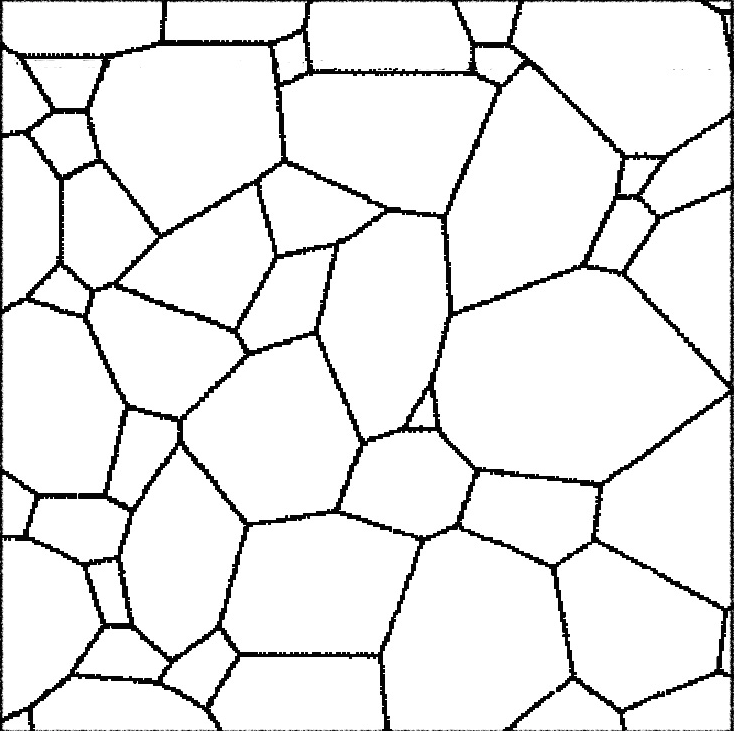}}

\subfloat[3a]{\includegraphics[width=5cm]{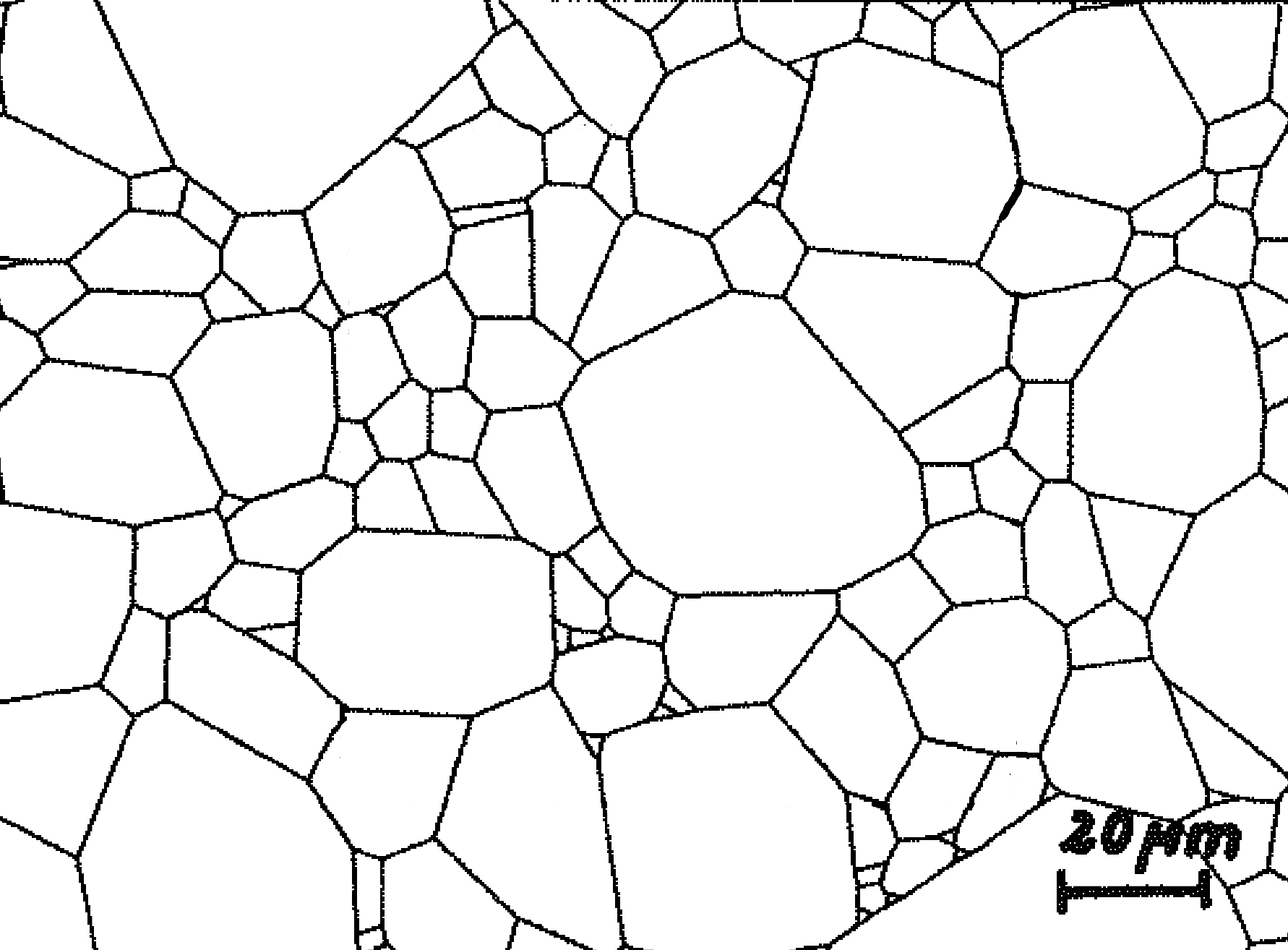}}
\subfloat[3b]{\includegraphics[width=4cm]{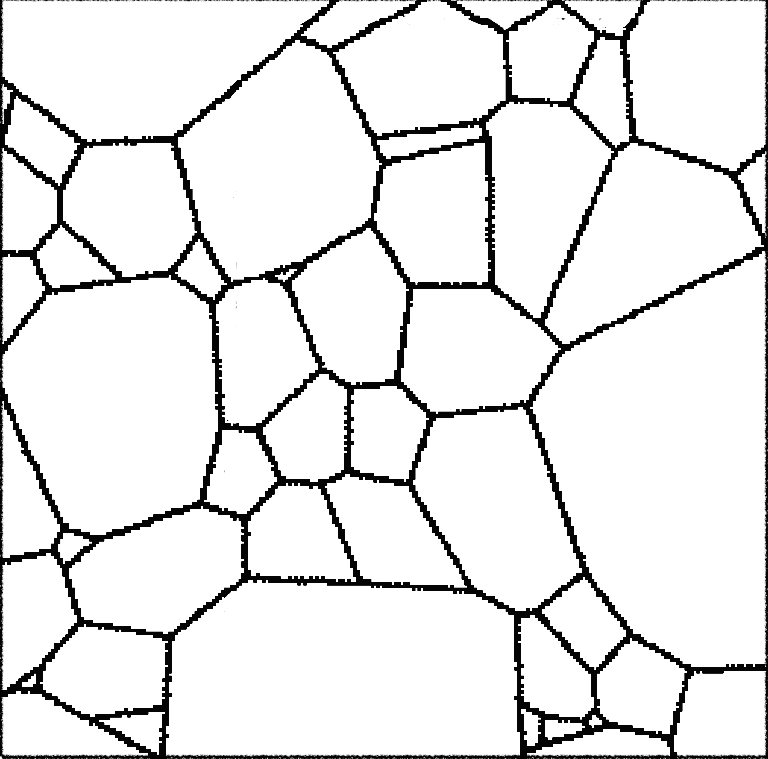}}
\caption{Schemes as planar tessellations of plane sections of alumina ceramics: preprocessing (a) Hahn\&Lorz (\cite{lorzhahn93}), (b) Cut of the plane sections with exactly 50 cells }
\label{fig:lorzhahn93}
\end{figure}

\begin{figure}[!h]
\centering
\includegraphics[width=8cm]{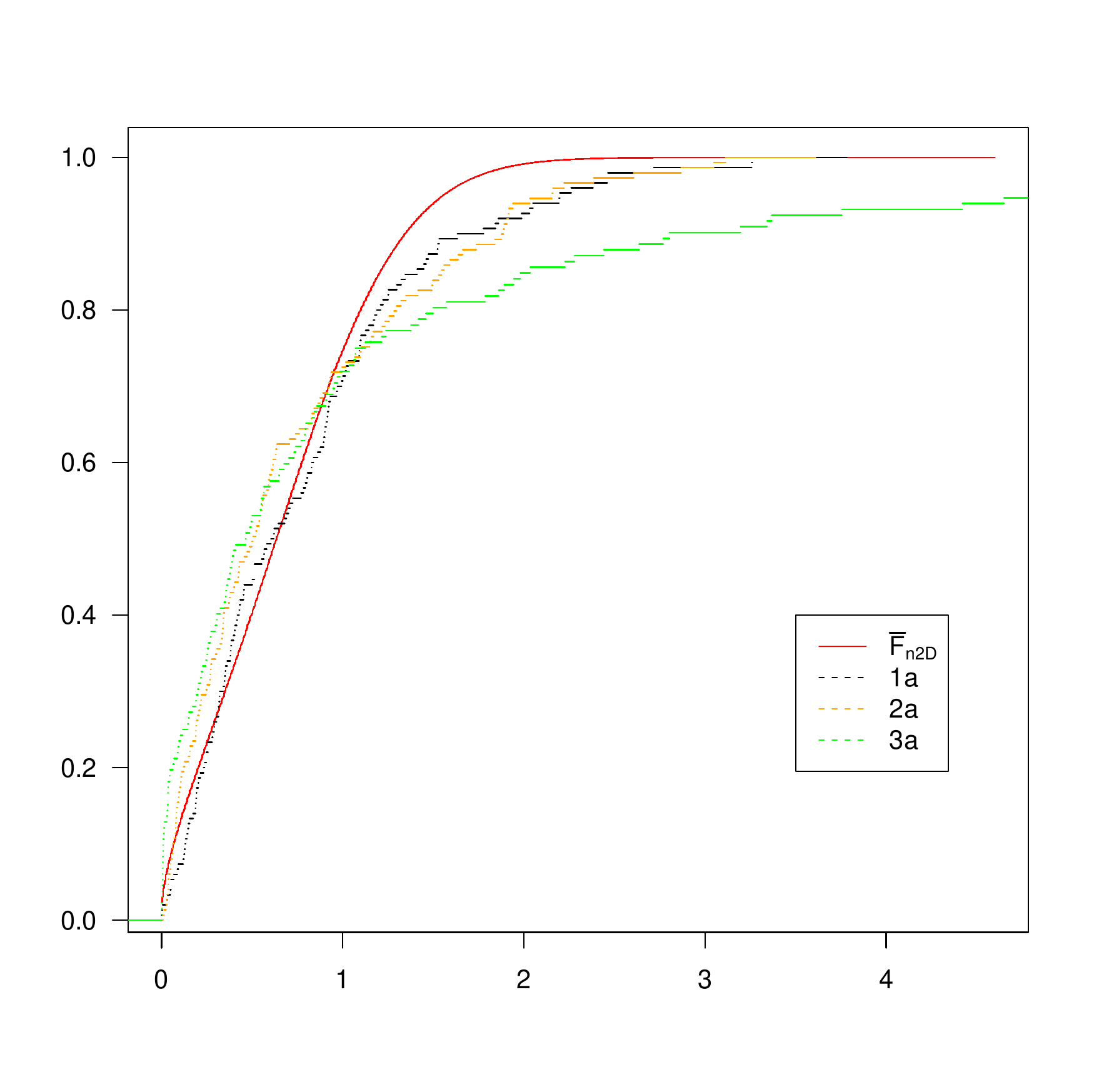}
\caption{Cumulative distribution function comparison of the cells area of the schemes of plane sections of alumina ceramics (Fig.\textbf{\ref{fig:lorzhahn93}} 1(a) black line, 2(a) yellow line, 3(a) green line) and of the 2D sectional Poisson-Voronoi cells area (red line)}
\label{fig:cdf50ctestall}
\end{figure}

\begin{figure}[!h]
\centering
\includegraphics[width=16cm]{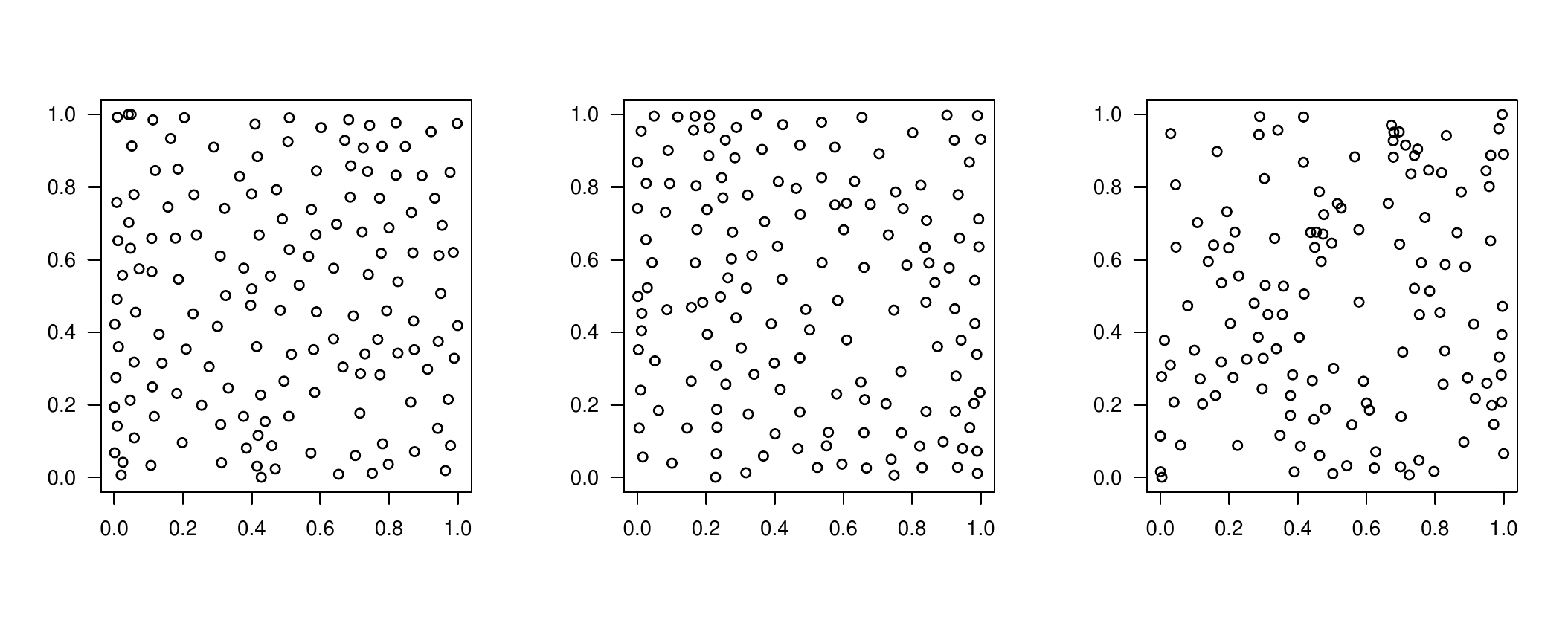}
\caption{From left to right centers of mass of the schemes of plane sections of alumina ceramics (Fig.\textbf{\ref{fig:lorzhahn93}} 1(a), 2(a), 3(a))  }
\label{fig:lorzpp50all}
\end{figure}

\begin{figure}[!h]
\centering
\includegraphics[width=16cm]{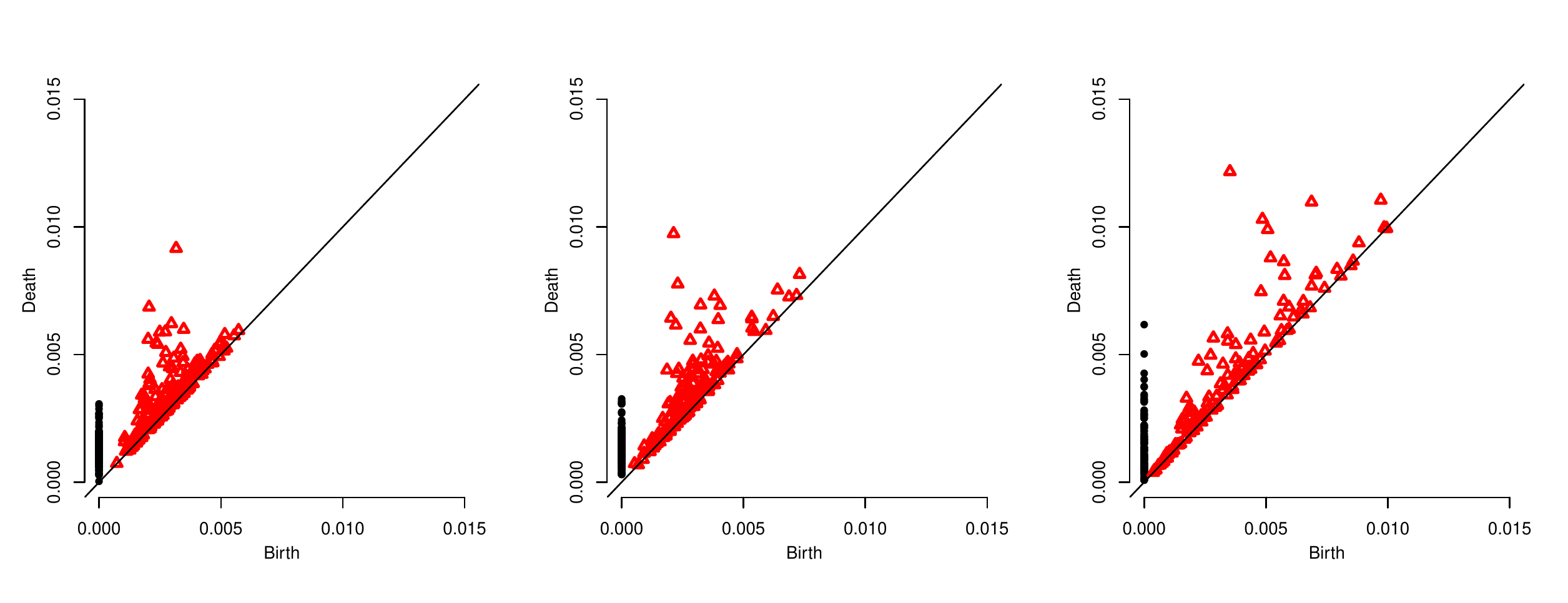}
\caption{From left to right persistence diagrams of the centers of mass of the schemes of plane sections of alumina ceramics (Fig.\textbf{\ref{fig:lorzhahn93}} 1(a), 2(a), 3(a)) }
\label{fig:perdiaglorzall}
\end{figure}

\begin{figure}[!h]
\centering
\includegraphics[width=16cm]{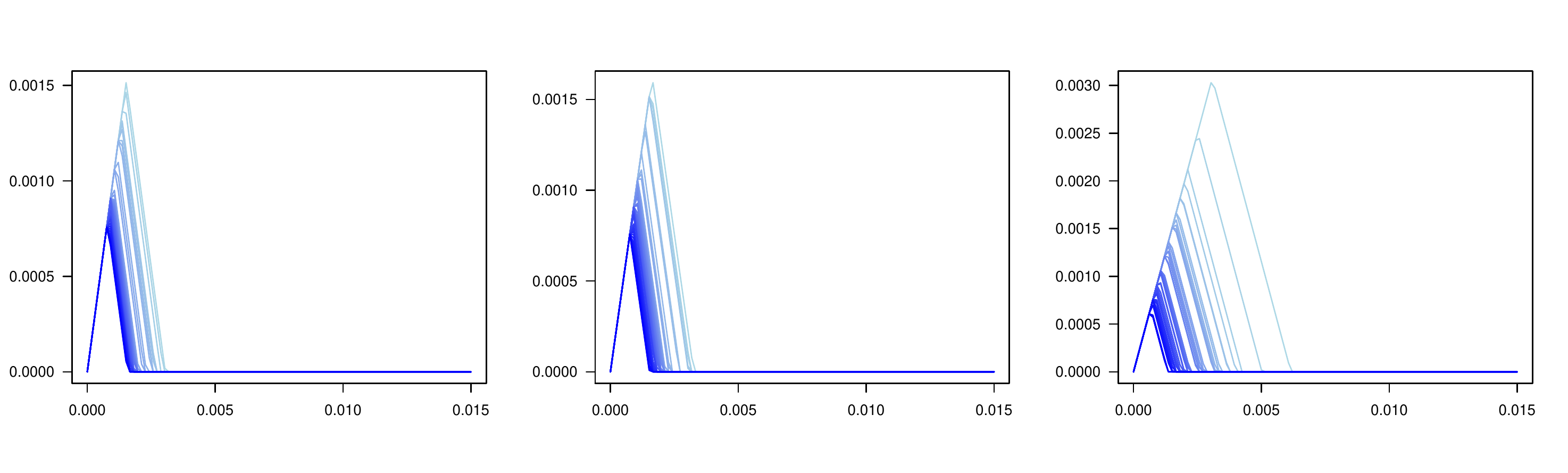}
\caption{From left to right persistence landscapes (connected components) of the schemes of plane sections of alumina ceramics (Fig.\textbf{\ref{fig:lorzhahn93}} 1(a), 2(a), 3(a)) }
\label{fig:landtest0all}
\end{figure}

\begin{figure}[!h]
\centering
\includegraphics[width=16cm]{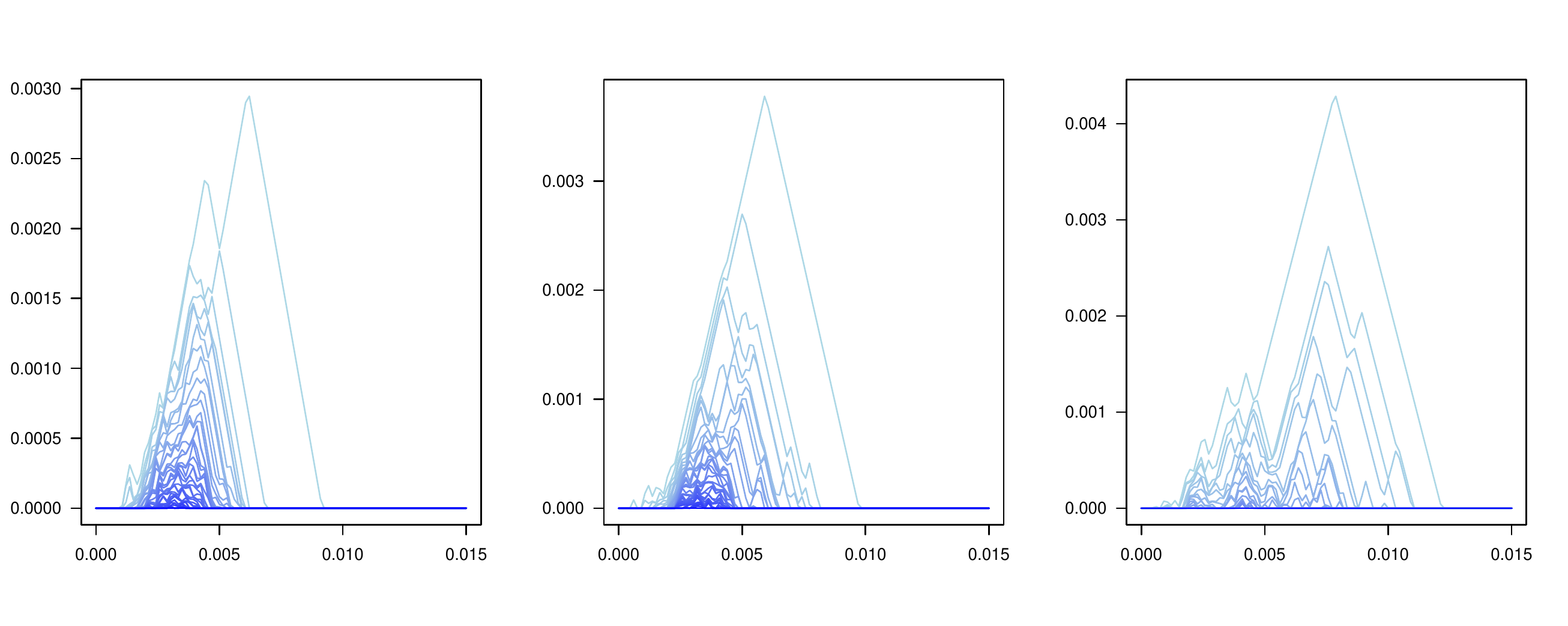}
\caption{From left to right persistence landscapes (holes) of the schemes of plane sections of alumina ceramics (Fig.\textbf{\ref{fig:lorzhahn93}} 1(a), 2(a), 3(a)) }\label{fig:landtest1all}
\end{figure}

\begin{figure}[!h]
\centering
\includegraphics[width=8cm]{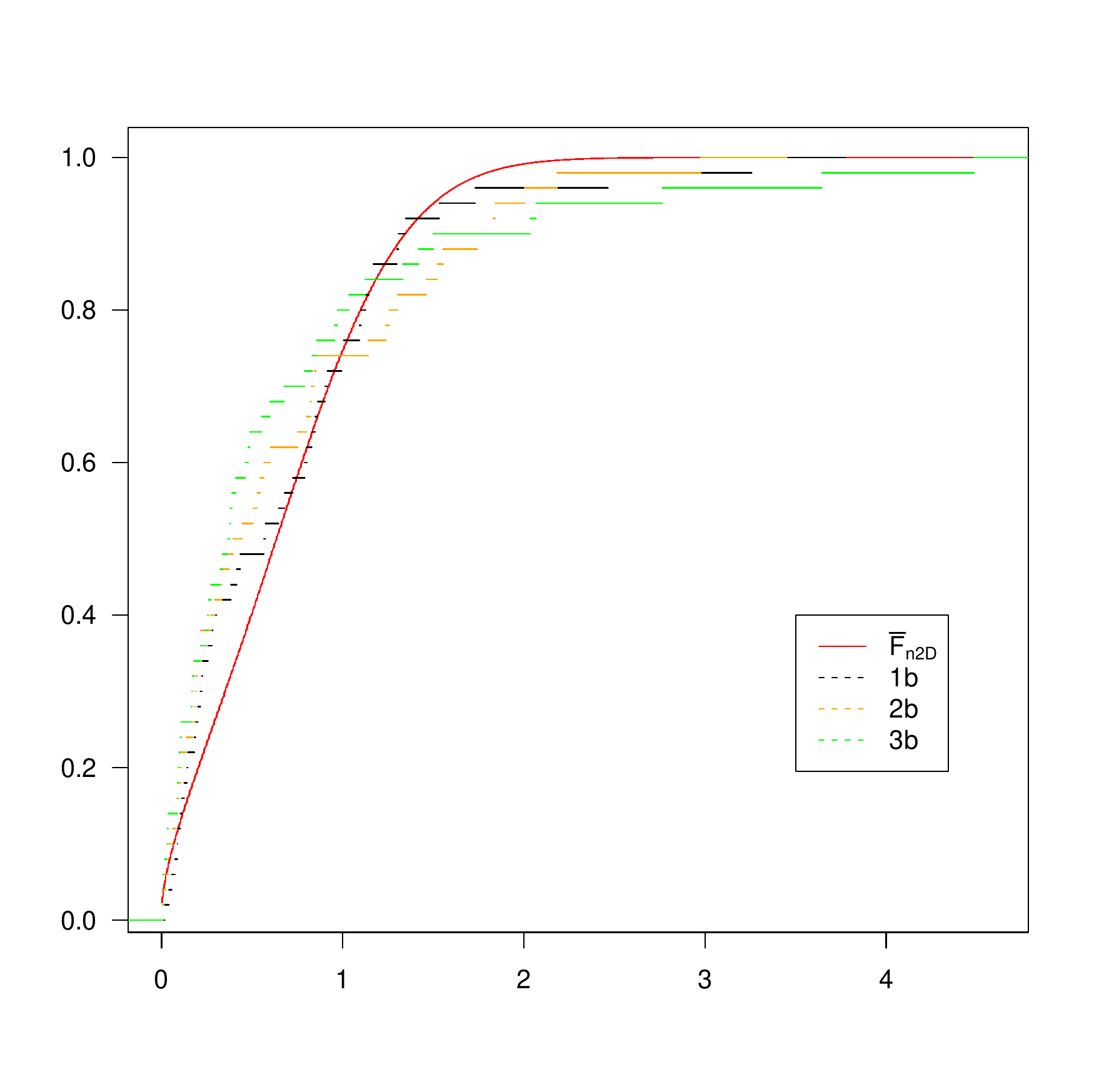}
\caption{Cumulative distribution function comparison of the cuts of the sections of alumina ceramics with exactly 50 cells (Fig.\textbf{\ref{fig:lorzhahn93}} 1(b) black line, 2(b) yellow line, 3(b) green line) and of the 2D sectional Poisson-Voronoi cells area conditioned on $N_{2D}=50$ (red line)}\label{fig:cdf50ctest}
\end{figure}

\begin{figure}[!h]
\centering
\includegraphics[width=16cm]{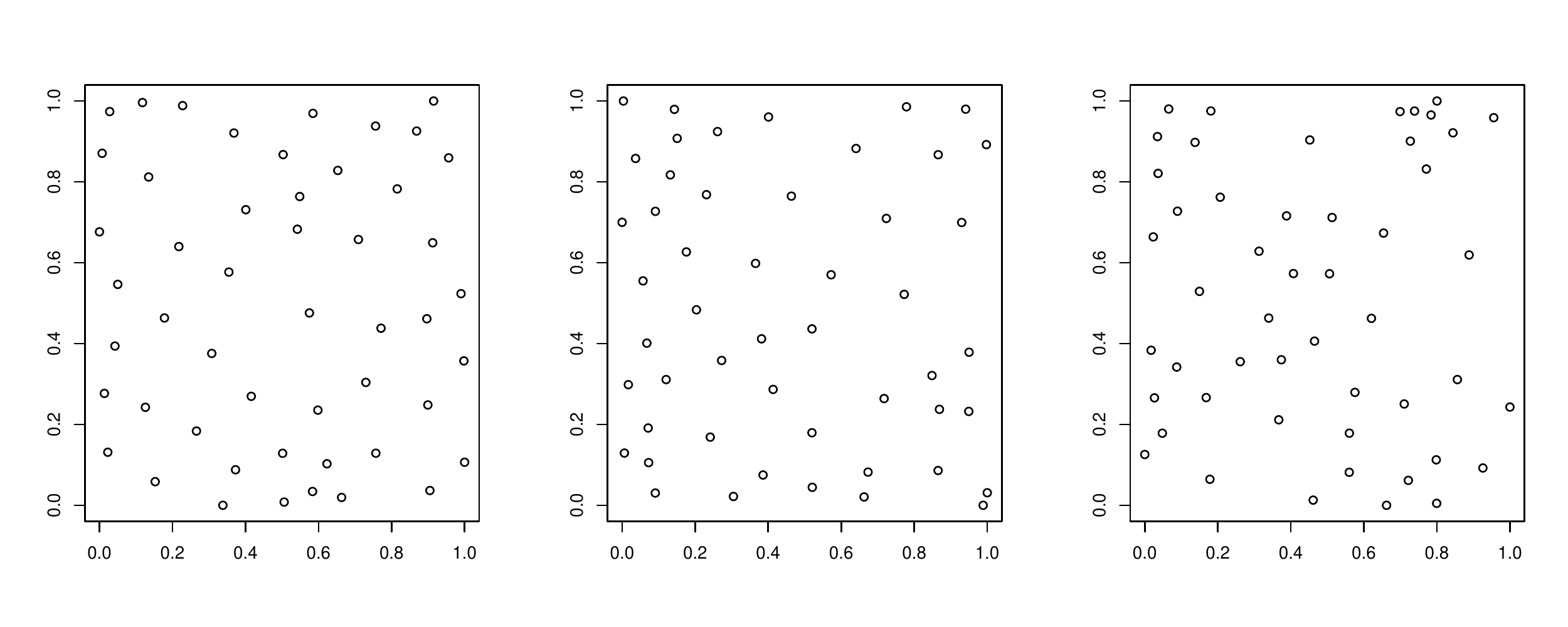}
\caption{From left to right centers of mass of the cuts of the sections of alumina ceramics with exactly 50 cells (Fig.\textbf{\ref{fig:lorzhahn93}} 1(b), 2(b), 3(b)) }\label{fig:lorzpp50}
\end{figure}

\begin{figure}[!h]
\centering
\includegraphics[width=16cm]{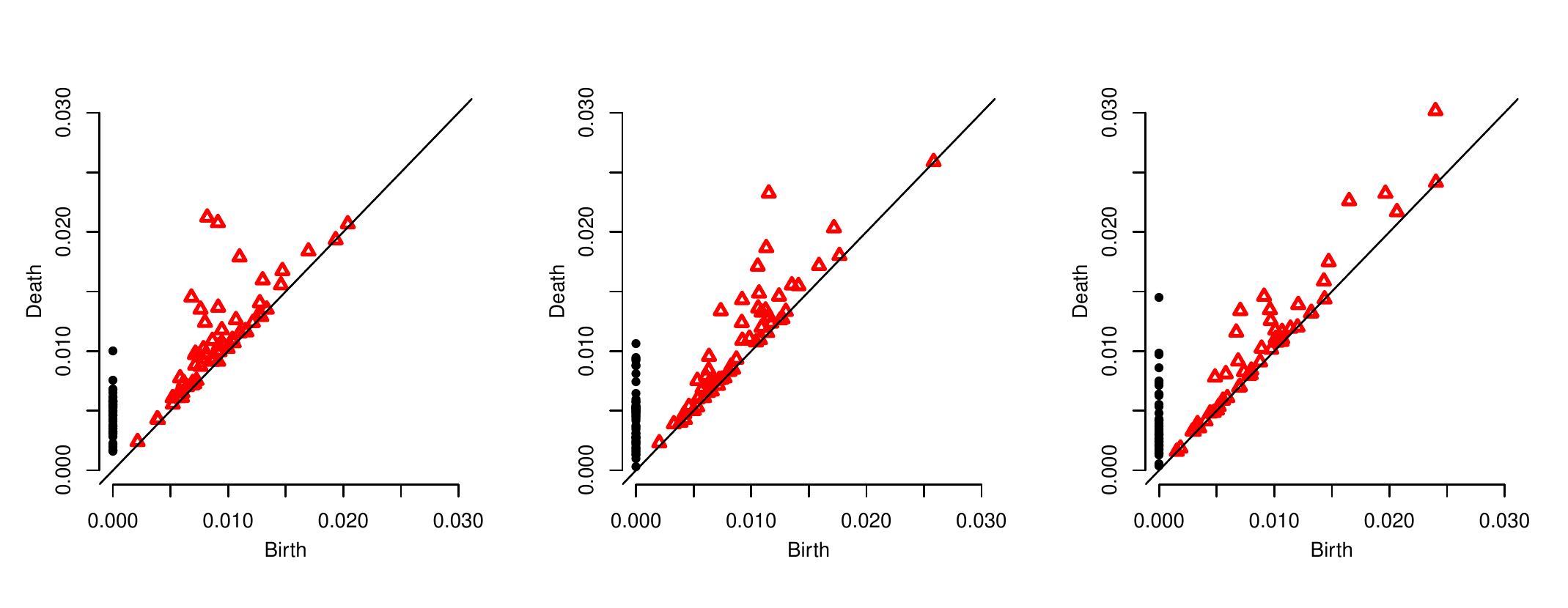}
\caption{From left to right persistence diagrams of the centers of mass of the cuts of the sections of alumina ceramics with exactly 50 cells (Fig.\textbf{\ref{fig:lorzhahn93}} 1(b), 2(b), 3(b))}
\label{fig:perdiaglorz}
\end{figure}

\begin{figure}[!h]
\centering
\includegraphics[width=16cm]{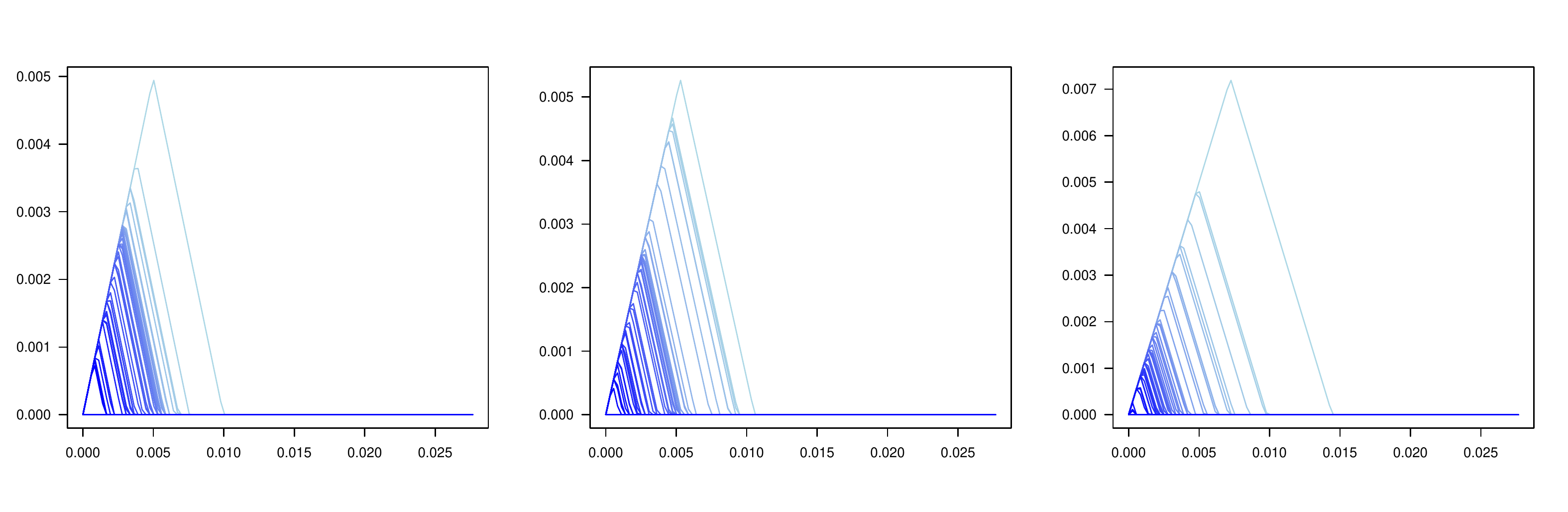}
 \caption{From left to right persistence landscapes (connected components) of the cuts of the sections of alumina ceramics with exactly 50 cells (Fig.\textbf{\ref{fig:lorzhahn93}} 1(b), 2(b), 3(b)) }
\label{fig:landtest0}
\end{figure}

\begin{figure}[!h]
\centering
\includegraphics[width=16cm]{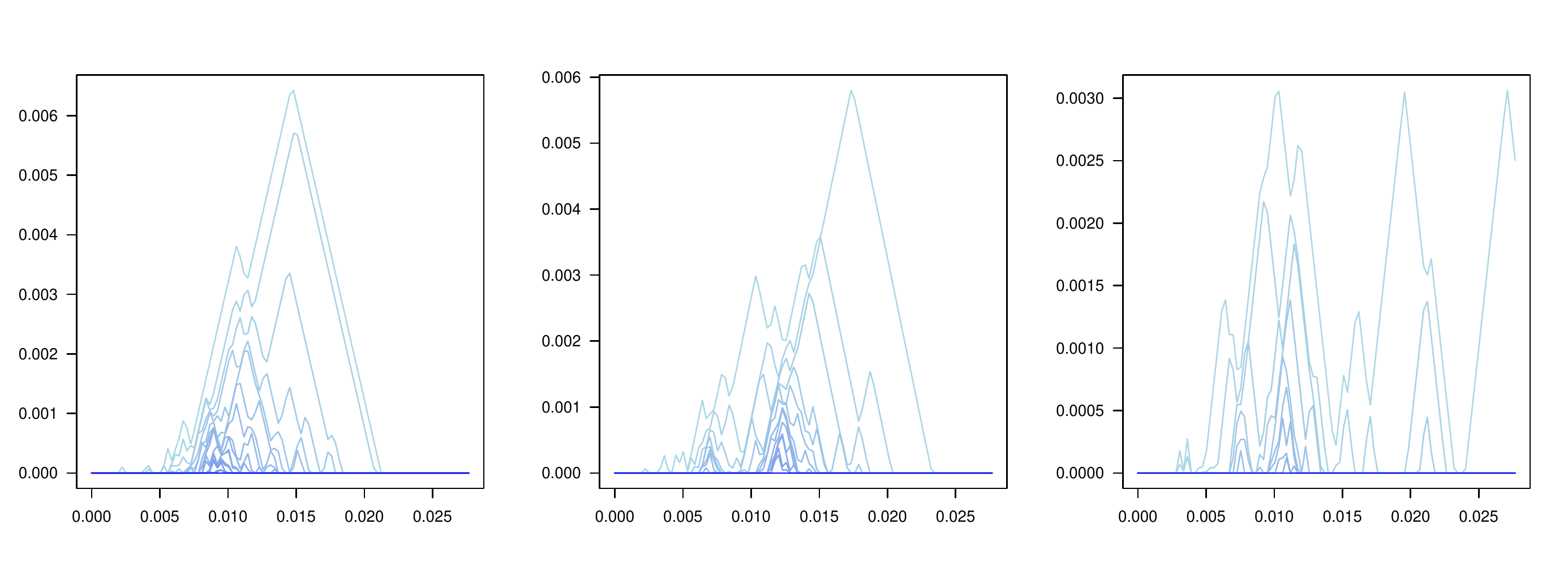}
 \caption{From left to right persistence landscapes (holes) of the cuts of the sections of alumina ceramics with exactly 50 cells (Fig.\textbf{\ref{fig:lorzhahn93}} 1(b), 2(b), 3(b)) }

\label{fig:landtest1}
\end{figure}

\begin{table}[!h]
\centering
\caption{Values of the different model tests for the schemes of plane sections of alumina ceramics (Fig.\textbf{\ref{fig:lorzhahn93}} 1(a), 2(a), 3(a)); p-value between brackets}
\subfloat[][\emph{1b}]
{\begin{tabular}{cl}
  \hline
$c$ & 0.848 (0.073)  \\
  $d$& 0.078 (0.710)\\
 $l_0$& 0.058 (0)\\
  $l_1 $& 0.019 (0) \\
   \hline
\end{tabular}}
\subfloat[][\emph{2b}]
{\begin{tabular}{cl}
  \hline
$c$ & 0.959 (0.002) \\
  $d$&  0.121 (0.172)\\
 $l_0$& 0.057 (0)\\
  $l_1 $& 0.028 (0)\\
   \hline
\end{tabular}}
\subfloat[][\emph{3b}]
{\begin{tabular}{cl}
  \hline
$c$ & 1.492 (0) \\
  $d$& 0.168 (0.006)\\
 $l_0$&  0.062 (0) \\
  $l_1 $& 0.018 (0)\\
   \hline
\end{tabular}}
\label{tab:lorztestall}
\end{table}

\begin{table}[!h]
\centering
\caption{Values of the different model tests for the cuts of the sections of alumina ceramics with exactly 50 cells (Fig.\textbf{\ref{fig:lorzhahn93}} 1(b), 2(b), 3(b)); p-value between brackets}
\subfloat[][\emph{1b}]
{\begin{tabular}{cl}
  \hline
$c$ & 0.931 (0.004)  \\
  $d$& 0.154 (0.014)\\
 $l_0$& 0.077 (0)\\
  $l_1 $& 0.041 (0) \\
   \hline
\end{tabular}}
\subfloat[][\emph{2b}]
{\begin{tabular}{cl}
  \hline
$c$ & 1.002 (0.0002) \\
  $d$&  0.172 (0.004)\\
 $l_0$& 0.116 (0)\\
  $l_1 $& 0.024 (0)\\
   \hline
\end{tabular}}
\subfloat[][\emph{3b}]
{\begin{tabular}{cl}
  \hline
$c$ & 1.328 (0) \\
  $d$& 0.248 (0)\\
 $l_0$&  0.137 (0) \\
  $l_1 $& 0.009 (0)\\
   \hline
\end{tabular}}
\label{tab:lorztest}
\end{table}

\clearpage
\section{Results and Discussion}
\label{sec:discussion}
This work provides a general setting for testing whether a microstructure is generated by a Poisson-Voronoi diagram, based on a cross section of the microstructure. Taking inspiration from previous work in this field, \cite{lorzkrawietz91,lorzhahn93}, we widen the testing framework proposing new model tests. In particular, we introduce test statistics using tools coming from different statistical branches like Goodness of Fit and Topological Data Analysis. We consider the situation with periodic boundary conditions, which is popular in materials science applications and without these conditions. Our approach is very general and can also be extended to test hypothesis for more complicated models describing the 3D structure based on a 2D section.

Being able to accept the Poisson-Voronoi model on the basis of 2D real metal sections means having complete probabilistic information on the underlying 3D structure. Based on this model, one can then perform mechanical experiments using virtual microstructures avoiding waste of material and discovering new interesting relations between microstructure features and mechanical properties much faster than possible using physical experiments.

Future developments involve testing of more general and less understood Voronoi structures for more complicated microstructures, such as Multi-level Voronoi diagrams. Another interesting direction is to consider fully data based approaches for analyzing 2D sections.  For instance, analyzing such a section using a persistence landscape  does not need the rigid restrictions on the geometry of the cells as present in the Poisson-Voronoi model. 
\section*{Acknowledgements}
This research was carried out under project number S41.5.14547b in the framework of the Partnership Program of the Materials innovation institute M2i (www.m2i.nl) and (partly) financed by the Netherlands Organisation for Scientific Research (NWO).
We also thank Vanessa Robins and Jeroen Spandaw for the inspiring discussion about Persistent Homology.

\end{document}